\documentclass[a4paper]{article}
\usepackage{amsfonts,amssymb,amsmath,amsthm,cite}
\usepackage{ucs} 
\usepackage[utf8x]{inputenc}
\usepackage{graphicx}
\usepackage[english]{babel}
\usepackage{slashed}
\usepackage{textcomp}
\usepackage{bm}
\usepackage{titlesec}
\usepackage{stackengine}

\usepackage{etoolbox}
\patchcmd{\thebibliography}{\section*}{\section}{}{}
\usepackage{mathtools}
\mathtoolsset{showonlyrefs}

\evensidemargin=\oddsidemargin
\newtheorem{theorem}{Theorem}
\newtheorem{lemma}{Lemma}
\newtheorem{remark}{Remark}
\newtheorem{defa}{Definition}
\newtheorem{col}{Corollary}

\begin{document}

\vspace{20mm}
\begin{center}
	\large{\textbf{Effective actions, cutoff regularization,}\\ 
		\textbf{quasi-locality, and gluing of partition functions}}
\end{center}
\vspace{2mm}
\begin{center}
	\large{\textbf{Aleksandr V. Ivanov}}
\end{center}
\begin{center}
	St. Petersburg Department
	of Steklov Mathematical Institute
	of Russian Academy of Sciences,\\
	27 Fontanka, St. Petersburg, 191023, Russia
\end{center}
\begin{center}
	Saint Petersburg State University,\\ 	7/9 Universitetskaya Emb., St. Petersburg, 199034, Russia
\end{center}
\begin{center}
	E-mail: regul1@mail.ru
\end{center}
\vspace{5mm}
\begin{flushright}
	\large{\textbf{\textit{ To the 100-th anniversary of N.M.Kalita}}}
\end{flushright}
\vspace{10mm}

\textbf{Abstract.} The paper studies a regularization of the quantum (effective) action for a scalar field theory in a general position on a compact smooth Riemannian manifold. As the main method, we propose the use of a special averaging operator, which leads to a quasi-locality and is a natural generalization of a cutoff regularization in the coordinate representation in the case of a curved metric. It is proved that the regularization method is consistent with a process of gluing of manifolds and partition functions, that is, with the transition from submanifolds to the main manifold using an additional functional integration. It is shown that the method extends to other models, and is also consistent with the process of multiplicative renormalization. Additionally, we discuss issues related to the correct introduction of regularization and the locality.

\vspace{2mm}

\newpage	
\tableofcontents

\section{Introduction}
\label{sk:sec:int}

Currently, the functional integration method \cite{sk-1,sk-2,3} is a widely used tool for studying quantum field models \cite{9,10} and, as a rule, implies a transition to an formal series expansion with respect to some auxiliary small parameter. Thus, in a suitable limit, the calculation of the "functional integral" is reduced to calculating the coefficients of the series, which can be represented in the form of well-studied Gaussian integrals. The use of such techniques is more of a necessity, since a complete mathematical description for the "functional integral" does not yet exist. Nevertheless, the method occupies a key position in modern mathematical and theoretical physics.

One of the problems associated with the study of the mentioned series is the presence of divergences, see \cite{105}. This is due to the fact that the coefficients of the series are constructed using Green's functions (fundamental solutions of a special kind, see \cite{Gelfand-1964,Vladimirov-2002}), nonlinear combinations of which can lead either to the appearance of non-integrable densities or to singular functionals at singular points. Regularization and elimination (renormalization) of such objects play an important role \cite{6,7}, since they allow not only to establish new mathematical relations or understand the internal structure of the coefficients, but also to give a clear physical meaning to the quantities under consideration.

Recently, the field of mathematics related to the functorial quantum field theory \cite{sk-3,sk-4,sk-5,sk-6,sk-7} has been actively developing. Despite the fact that at the moment there are many studied models \cite{sk-8,sk-9,sk-10}, examples related to the consideration of well-known quantum field densities on Riemannian manifolds are quite rare \cite{sk-11,sk-12,sk-4,sk-15,sk-13,sk-14,sk-16}. One of the main problems in this area is the task of regularization and elimination of divergences. In particular, it is unclear how to construct a regularization that would work in all dimensions and be consistent with the process of "gluing" of partition functions, that is, with the transition from two smooth Riemannian submanifolds to a third smooth Riemannian manifold by connecting along a common boundary.

This paper describes the construction of a universal quasi-local regularization based on a deformation of a classical action by adding a deforming operator to the term responsible for the interaction. It is shown that the proposed regularization works in all dimensions and is consistent with the process of gluing manifolds and partition functions. It is important to note that in this paper the described procedure is a generalization of a cutoff regularization in the coordinate representation for the case of a smooth Riemannian manifold with a boundary. This approach was first proposed in \cite{34} and then significantly improved in subsequent works \cite{Iv-2024-1,Ivanov-Akac,Ivanov-Kharuk-2020,Ivanov-Kharuk-20222,Ivanov-Kharuk-2023,Iv-Kh-2024,Kh-2024}. The peculiarity of this approach is the transparency of the spectral meaning, the connection with the averaging operator \cite{Ivanov-2022,Iv-2024}, the fulfillment of an "applicability condition", see also\footnote{The paper \cite{sk-b-20} discusses the possibility of presenting the answer for averages from \cite{Iv-2024} in elementary functions.} \cite{Iv-2024}, related to the non-negativity of the spectrum, a quasi-locality, as well as the possibility of performing explicit loop calculations.

At the same time, it is important to pay attention to the fact that other approaches to regularization are used when working with standard quantum field models in flat space.  Among the most popular types, an alternative version of a cutoff regularization can be distinguished \cite{w6,w7,w8,Khar-2020,sk-b-19}, implicit regularization \cite{chi-0,chi-1,chi-2}, Pauli--Villars regularization \cite{Pauli-Villars}, and also dimensional one \cite{19,555} and using higher covariant derivatives \cite{Bakeyev-Slavnov,29-st}. At the same time, the last two play an essential role in theoretical physics and are currently indispensable from the point of view of multi-loop calculations. However, all these regularizations are poorly compatible with the use of classical mathematical physics, and consistency with the process of gluing manifolds and partition functions is either absent or unclear.

Moreover, the approach to regularization proposed in this paper seems to be distinctive, since it not only has a strict mathematical formulation that preserves the applicability of classical mathematical physics, but also includes a large number of explicit loop calculations already performed within the framework of standard quantum field models. In particular, consistency has been demonstrated in various theories with the results obtained using other regularization methods widely used in physics and having problems on Riemannian manifolds. Thus, it can be argued that the proposed approach is consistent not only with the process of gluing of manifolds, but also, paying attention to the conceptual component, is consistent with the gluing of two fields of science, mathematical and theoretical physics, in the context of quantum field theory.

\vspace{2mm}
\textbf{The structure of the work.}

\vspace{2mm}
\textbf{Section \ref{sk:sec:def}} is devoted to the description of the main objects and their properties. In particular, we give the definitions of a classical action of theory under consideration, a Laplace operator, a basic manifold with boundary and its deformed submanifolds. Next, useful relations are formulated for Green's functions constructed for the Dirichlet problem with zero boundary conditions, as well as some characteristics in terms of local and global properties. In the last section, we introduce a quantum (effective) action in the form of an formal series in powers of a small parameter $\sqrt{\hbar}$ and in the form of a formal "functional integral".

\vspace{2mm}
\textbf{Section \ref{sk:sec:reg}} contains a description of the process of regularization of the quantum action and a partition function by deformation of the classical action. The first part of the section is devoted to the description of the already well-known method of a cutoff regularization in coordinate representation in Euclidean space. Next, the approach is generalized to the case of a smooth Riemannian manifold with a boundary. In the last part, we formulate the main recipe for the introduction of the regularization and also discuss a quasi-locality. The definition of the partition function is given as well.

\vspace{2mm}
\textbf{Section \ref{sk:sec:app}} is devoted to proving the consistency of the proposed regularization with respect to the process of gluing of the partition functions (quantum actions). The section contains a definition of the concept of "partition functions gluing", the proof of the main theorem, as well as several important consequences. Among the most significant results of this section are the following four items.
\begin{itemize}
	\item \textbf{Lemma \ref{sk-l-11}} contains the derivation of deformed "gluing" relations of Green's functions on the submanifolds.
	\item In \textbf{Theorem\ref{sk-t-1}}, the consistency of the proposed regularization with the process of "gluing" for the partition functions (and quantum actions) is proved.
	\item In \textbf{Corollary \ref{sk-c-10}}, the possibility of an operator structure of the coupling coefficients is discussed.
	\item In \textbf{Corollary \ref{sk-c-11}}, we discuss the issue of consistency with the multiplicative renormalization process.
\end{itemize}

\vspace{2mm}
\textbf{Section \ref{sk:sec:dis}} contains various useful comments, remarks and acknowledgements. Separately, it is worth mentioning the discussion of applicability and extension to a wider class of models.

\vspace{2mm}
The text is organized in such a way as to be understandable to both mathematicians and theoretical physicists. This is manifested, in particular, in the use of a mathematical style of presentation, which includes definitions of objects in a form familiar to physicists. For example, when introducing the classical action, not the language of differential forms is used, but the standard notation in local coordinates. This approach allows us to maintain mathematical rigor, as well as prepare the reader for subsequent calculations and possible studies of divergences in specific models in the future.

\section{The main definitions}
\label{sk:sec:def}

\subsection{Classical action}
\label{sk:sec:def:1}

Let $\mathcal{M}$ be a smooth compact orientable n-dimensional Riemannian manifold, see \cite{sk-b-1,Nakahara:2003nw}, with a boundary $\partial\mathcal{M}=Y$, which may be absent, that is, the situation $\partial\mathcal{M}=\emptyset$ is possible. Next, we introduce the atlas $\mathcal{A}=\{(U_\alpha,\varphi_\alpha)\}_{\alpha\in I}$, consisting of a finite\footnote{This is possible because the manifold is compact. In general, finiteness is not a necessary condition for the constructions under consideration and is assumed only for convenience.} number of coordinate charts. Note that the smooth function $\varphi_\alpha(\cdot)$ maps the set $U_\alpha\subset\mathcal{M}$ to $V_\alpha\subset\mathbb{R}^n$, and the element $p\in U_\alpha$ to the element $x\in V_\alpha$. The components $x=(x^1,\ldots,x^n)$ will be notated by the indices $\mu,\nu\in\{1,\ldots,n\}$. At the same time, for the repeated\footnote{Since the case with a non-trivial metric is considered, the "upper" and "lower" indices have a fundamental difference. There will always be one upper index and one lower index in the summation by the repeating indices.}  indices (only $\mu$ or $\nu$) we have the summation over all possible values.

It is clear that in each card $(U_\alpha,\varphi_\alpha)\in\mathcal{A}$, we can write out the metric tensor $g_\alpha^{\mu\nu}(x)$ and its determinant $g_\alpha(x)$. Additionally, we define the auxiliary partition of unity, that is, the set of functions $\{\chi_\alpha\in C^\infty(\mathcal{M},[0,1])\}_{\alpha\in I}$, having the properties
\begin{equation}\label{sk-2-4}
\sum_{\alpha\in I}\chi_\alpha(p)=1\,\,\,\mbox{for all}\,\,\,p\in\mathcal{M},\,\,\,
\mathop{\mathrm{supp}}_{\mathcal{M}}(\chi_\alpha)\subset U_\alpha\,\,\,\mbox{for all}\,\,\,\alpha\in I.
\end{equation}
Next, consider the function $\phi\in C^\infty(\mathcal{M},\mathbb{R})$. In this case, within each neighborhood $U_\alpha$, the mapping $\phi(\cdot)$ corresponds to a function $\phi_\alpha(x)\in C^\infty(V_\alpha,\mathbb{R})$ such that $\phi(p)=\phi_\alpha(\varphi_\alpha(p))$ for all $p\in U_\alpha$. Then we introduce two functionals, which are defined by the following expressions
\begin{equation}\label{sk-2-1}
S_0[\phi;\mathcal{M}]=\frac{1}{2}
\sum_{\alpha\in I}\int_{V_\alpha}\mathrm{d}^nx\,g^{1/2}_\alpha(x)
\bigg(g^{\mu\nu}_\alpha(x)
\Big(\partial_{x^\mu}\phi^{\phantom{1}}_\alpha(x)\Big)
\Big(\partial_{x^\nu}\phi^{\phantom{1}}_\alpha(x)\Big)+m^2\phi^2_\alpha(x)\bigg)\chi_\alpha^{\phantom{1}}(\varphi_\alpha^{-1}(x)),
\end{equation}
\begin{equation}\label{sk-2-5}
	S_{\mathrm{int}}[\phi;\mathcal{M}]=
	\sum_{k=3}^{+\infty}
	\sum_{\alpha\in I}\int_{V_\alpha}\mathrm{d}^nx\,t_k^{\phantom{'}}g^{1/2}_\alpha(x)
	\phi^k_\alpha(x)\chi_\alpha^{\phantom{1}}(\varphi_\alpha^{-1}(x)).
\end{equation}
Here the coefficients\footnote{The case of nonzero and real $t_1,t_2$ can be reduced to the form \eqref{sk-2-5} by using a shift and redefining the mass parameter. In this case, $m^2$ will not necessarily be positive. Generalization to an arbitrary case is discussed in Section \ref{sk:sec:app-3}.} $t_k$ are called coupling constants, and $m\geqslant0$ plays the role of a mass parameter\footnote{We will assume that $m^2>0$, see the discussion after Lemma \ref{sk-l-6}.}.
\begin{remark}\label{sk-r-1}
The set of parameters $\{t_k\in\mathbb{C}\}_{k=3}^{+\infty}$ cannot take arbitrary values, this is due to the fact that the sum in \eqref{sk-2-5} may lose convergence. Next, we will assume that the coefficients correspond to a convergent series. Moreover, we assume that the real part of the functional $S_{\mathrm{int}}[\,\cdot\,,\mathcal{M}]$ is bounded from below, that is, there exists $C\in\mathbb{R}$, for which
\begin{equation}\label{sk-2-2}
\Re(S_{\mathrm{int}}[\phi,\mathcal{M}])\geqslant C\,\,\,\mbox{for all}\,\,\,\phi\in C^\infty(\mathcal{M},\mathbb{R}).
\end{equation}
\end{remark}
\noindent Here we have three common examples.
\begin{itemize}
	\item Cubic model: $t_3\in i\mathbb{R}\setminus\{0\}$ and $t_k=0$ for all $k>3$, see \cite{sk-b-2,sk-b-3,34}.
	\item Quartic model: $t_3,t_4\in\mathbb{C}$, $\Re(t_4)>0$ and $t_k=0$ for all $k>4$, see \cite{Iv-2024-1,29-3,29-4}.
	\item Sextic model: $t_3,t_4,t_5,t_6\in\mathbb{C}$, $\Re(t_6)>0$ and $t_k=0$ for all $k>6$, see \cite{Kh-2024,NK-1,NK-2,NK-3,NK-4,NK-5}.
\end{itemize}
\begin{defa}\label{sk-d-6}
The classical action $S_{\mathrm{cl}}[\,\cdot\,;\mathcal{M}]$  of scalar theory on the manifold $\mathcal{M}$ is equal to the sum of two functionals $S_0[\,\cdot\,,\mathcal{M}]+S_{\mathrm{int}}[\,\cdot\,;\mathcal{M}]$.
\end{defa}

\subsection{Submanifolds}
\label{sk:sec:def:4}
Let the assumptions of Section \ref{sk:sec:def:1} be valid. Consider a smooth closed\footnote{A compact manifold without a boundary, that is, $\partial\Sigma=\emptyset$.} submanifold $\Sigma$ of codimension 1, that is, $\dim\Sigma=n-1$, which divides the manifold $\mathcal{M}$ into two smooth compact orientable $n$-dimensional Riemannian manifolds, which will be denoted by $\mathcal{M}_l$ and $\mathcal{M}_r$. In this case, the set $Y=\partial\mathcal{M}$ splits into two disconnected parts $Y=Y_l\cup Y_r$ in such a way that
\begin{equation}\label{sk-4-1}
\partial\mathcal{M}_l=Y_l\cup\Sigma,\,\,\,\partial\mathcal{M}_r=Y_r\cup\Sigma.
\end{equation}
A schematic view is shown in Figure \ref{pic:sk-1}. Note that when gluing the manifold $\mathcal{M}$ back together from the parts $\mathcal{M}_l$ and $\mathcal{M}_r$, one "extra" boundary should be removed. This procedure is notated as $\mathcal{M}=\mathcal{M}_l\cup_{\Sigma}\mathcal{M}_r$.
\begin{defa}\label{sk-d-8}
Let the number $1/\Lambda_1$ denote the minimum of the set $\{d(p,q):p\in\Sigma,q\in Y\}$, where the function $d(\,\cdot\,,\,\cdot\,)$ denotes the geodesic distance.
\end{defa}
\begin{remark}\label{sk-r-2}
It is further assumed that $\Lambda_1<+\infty$.
\end{remark}

Then, for the obtained sets, we introduce the corresponding atlases and unit partitions. Considering the fact that the smooth case is studied in the work, the mentioned objects can be obtained by narrowing of domains. Indeed, let us demonstrate the process using the example of $\mathcal{M}_l$. First, we define a set of indices
\begin{equation}\label{sk-4-2}
	I_l=\{\alpha\in I:U_\alpha\cap\mathcal{M}_l\neq\emptyset\},
\end{equation}
then, as an atlas $\mathcal{A}_l$ and a unit partition, we choose
\begin{equation}\label{sk-4-3}
\mathcal{A}_l^{\phantom{1}}=\{(U_{l,\alpha}^{\phantom{1}},\varphi_{l,\alpha}^{\phantom{1}})\}_{\alpha\in I_l}^{\phantom{1}},\,\,\,
\{V_{l,\alpha}^{\phantom{1}}\}_{\alpha\in I_l}^{\phantom{1}},
\,\,\,\mbox{and}\,\,\,
\{\chi_{l,\alpha}^{\phantom{1}}\}_{\alpha\in I_l}^{\phantom{1}},
\end{equation}
where
\begin{equation}\label{sk-4-4}
U_{l,\alpha}^{\phantom{1}}=U_{\alpha}^{\phantom{1}}\cap\mathcal{M}_l^{\phantom{1}}
,\,\,\,
\varphi_{l,\alpha}^{\phantom{1}}=\varphi_{\alpha}^{\phantom{1}}\big|_{\mathcal{M}_l}^{\phantom{1}}
,\,\,\,
V_{l,\alpha}^{\phantom{1}}=\varphi_{l,\alpha}^{\phantom{1}}\big(U_{l,\alpha}^{\phantom{1}}\big)
=\varphi_{\alpha}^{\phantom{1}}\big(U_{l,\alpha}^{\phantom{1}}\big)
,\,\,\,
\chi_{l,\alpha}^{\phantom{1}}=\chi_{\alpha}^{\phantom{1}}\big|_{\mathcal{M}_l}^{\phantom{1}}.
\end{equation}
Similarly, we obtain the corresponding atlas and partitions for the submanifold $\mathcal{M}_r$. They are notated as follows
\begin{equation}\label{sk-4-5}
I_r^{\phantom{1}},\,\,\,\mathcal{A}_r^{\phantom{1}}=\{(U_{r,\alpha}^{\phantom{1}},\varphi_{r,\alpha}^{\phantom{1}})\}_{\alpha\in I_r}^{\phantom{1}},\,\,\,
\{V_{r,\alpha}^{\phantom{1}}\}_{\alpha\in I_r}^{\phantom{1}},\,\,\,
\{\chi_{r,\alpha}^{\phantom{1}}\}_{\alpha\in I_r}^{\phantom{1}}.
\end{equation}
Continuing, we obtain the corresponding sets and functions for the boundaries of $Y_l$, $Y_r$, and $\Sigma$. Since the process is completely similar, we pay attention only to the choice of designations
\begin{align}
Y_l^{\phantom{1}}&\to \hat{I}_l^{\phantom{1}},\hat{U}_{l,\alpha}^{\phantom{1}},
\hat{\varphi}_{l,\alpha}^{\phantom{1}},
\hat{V}_{l,\alpha}^{\phantom{1}},
\hat{\chi}_{l,\alpha}^{\phantom{1}},\\
Y_r^{\phantom{1}}&\to \hat{I}_r^{\phantom{1}},\hat{U}_{r,\alpha}^{\phantom{1}},
\hat{\varphi}_{r,\alpha}^{\phantom{1}},
\hat{V}_{r,\alpha}^{\phantom{1}},
\hat{\chi}_{r,\alpha}^{\phantom{1}},\\
\Sigma&\to \hat{I}_\Sigma^{\phantom{1}},\hat{U}_{\Sigma,\alpha}^{\phantom{1}},
\hat{\varphi}_{\Sigma,\alpha}^{\phantom{1}},
\hat{V}_{\Sigma,\alpha}^{\phantom{1}},
\hat{\chi}_{\Sigma,\alpha}^{\phantom{1}}.
\end{align}
\begin{defa}\label{sk-d-7}
A deformation $\mathcal{M}_\Lambda$ of the manifold $\mathcal{M}$ is equal to a submanifold consisting of elements
\begin{equation}\label{sk-4-6}
\mathcal{M}_\Lambda=\{p\in\mathcal{M}:d(p,q)\geqslant1/\Lambda\,\,\,\mbox{for all}\,\,\,q\in\partial\mathcal{M}\}.
\end{equation}
\end{defa}
\begin{col}\label{sk-c-3}
For $\Lambda\to+\infty$ the submanifold $\mathcal{M}_\Lambda$ tends to $\mathcal{M}$.
\end{col}
\noindent Additionally introducing into consideration the deformations $\mathcal{M}_{\Lambda,l}$ and $\mathcal{M}_{\Lambda,r}$ for mentioned above manifolds, fix the corresponding notations for them
\begin{align}
\mathcal{M}_\Lambda^{\phantom{1}}&\to I_{\Lambda}^{\phantom{1}},U_{\Lambda,\alpha}^{\phantom{1}},
\varphi_{\Lambda,\alpha}^{\phantom{1}},
V_{\Lambda,\alpha}^{\phantom{1}},
\chi_{\Lambda,\alpha}^{\phantom{1}},\\
\mathcal{M}_{l,\Lambda}^{\phantom{1}}&\to I_{\Lambda,l}^{\phantom{1}},U_{\Lambda,l,\alpha}^{\phantom{1}},
\varphi_{\Lambda,l,\alpha}^{\phantom{1}},
V_{\Lambda,l,\alpha}^{\phantom{1}},
\chi_{\Lambda,l,\alpha}^{\phantom{1}},\\
\mathcal{M}_{r,\Lambda}^{\phantom{1}}&\to I_{\Lambda,r}^{\phantom{1}},U_{\Lambda,r,\alpha}^{\phantom{1}},
\varphi_{\Lambda,r,\alpha}^{\phantom{1}},
V_{\Lambda,r,\alpha}^{\phantom{1}},
\chi_{\Lambda,r,\alpha}^{\phantom{1}}.
\end{align}
Schematically deformed submanifolds are shown in Figure \ref{pic:sk-1}.

\begin{figure}[h]
	\center{\includegraphics[width=0.62\linewidth]{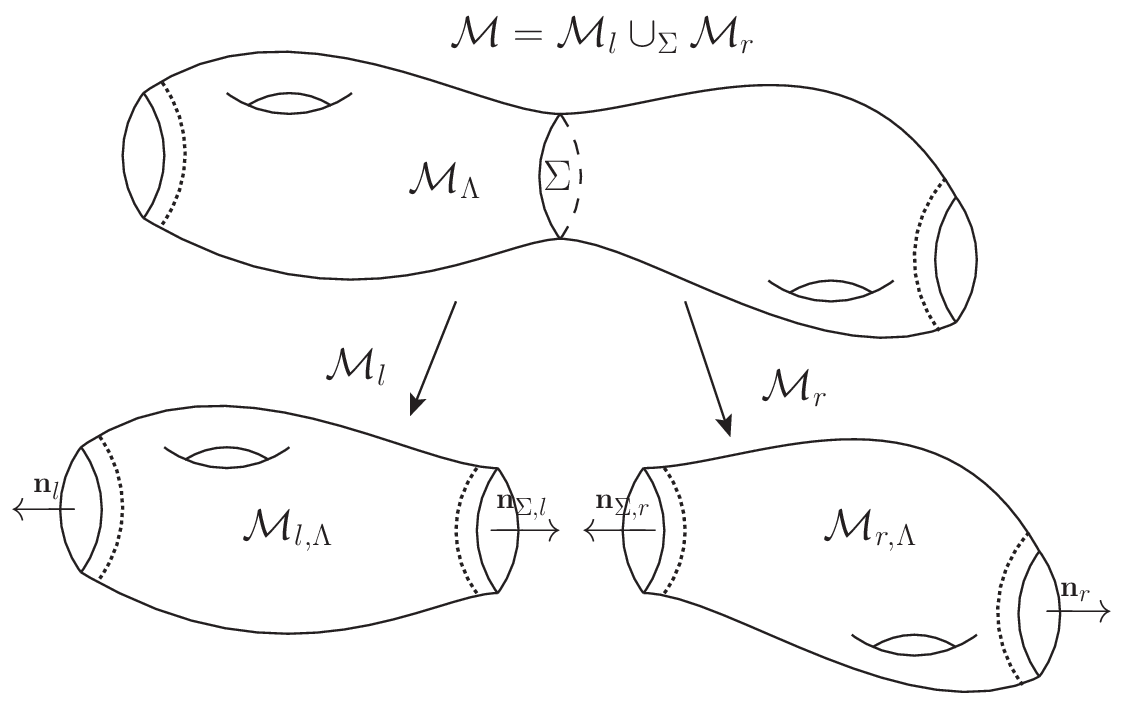}}
	\caption{The upper part we show the manifold $\mathcal{M}$ with the boundaries $Y_l$ on the left and $Y_r$ on the right, as well as the truncated (deformed) manifold $\mathcal{M}_\Lambda$, which differs from the original by the cut-off neighborhoods of the boundary with a thickness of $1/\Lambda$. Below we have two submanifolds $\mathcal{M}_{l}$ and $\mathcal{M}_{r}$, obtained by cutting the manifold $\mathcal{M}$ along $\Sigma$. Similarly, we obtain the deformed submanifolds  $\mathcal{M}_{l,\Lambda}$ and $\mathcal{M}_{r,\Lambda}$, as well as the corresponding vectors of external normals. }
	\label{pic:sk-1}
\end{figure}

\subsection{Operators and their properties}
\label{sk:sec:def:2}

In this section, we consider the introduction of a more concise entry for the classical action, see \eqref{sk-2-1} and \eqref{sk-2-5}, and also recall the basic properties of the operators included in it.
\begin{defa}\label{sk-d-9} Let $i\in\{l,r\}$. The symbols $\mathbf{n}_{i,\alpha}^\nu(x)$ and $\mathbf{n}_{\Sigma,i,\alpha}^\nu(x)$ denote the normalized external vectors of the normals to the boundaries $Y_i$ and $\Sigma$ of the submanifold $\mathcal{M}_i$ in the coordinate charts $(\hat{U}_{i,\alpha},\hat{\varphi}_{i,\alpha})$ and $(\hat{U}_{\Sigma,\alpha},\hat{\varphi}_{\Sigma,\alpha})$, respectively.
\end{defa}
\begin{col}\label{sk-c-4}
Taking into account Definition \ref{sk-d-9}, the relation $\mathbf{n}_{\Sigma,l,\alpha}^\nu(x)=-\mathbf{n}_{\Sigma,r,\alpha}^\nu(x)$ is correct.
\end{col}
Next, we integrate the first term in formula \eqref{sk-2-1} by parts and use the auxiliary definitions, then we get
\begin{multline}\label{sk-3-2}
S_0[\phi;\mathcal{M}]=\frac{1}{2}\sum_{\alpha\in I}\int_{V_\alpha}\mathrm{d}^nx\,g^{1/2}_\alpha(x)
\Big(\phi^{\phantom{1}}_\alpha(x)A_\alpha^{\phantom{1}}(x)\phi^{\phantom{1}}_\alpha(x)\Big)\chi_\alpha^{\phantom{1}}\big(\varphi_\alpha^{-1}(x)\big)+\\+\frac{1}{2}
\sum_{\alpha\in \hat{I}_l}\int_{\hat{V}_{l,\alpha}}\mathrm{d}^{n-1}x\,g^{1/2}_\alpha(x)
\phi^{\phantom{1}}_\alpha(x)
\Big(\mathbf{n}^\nu_{l,\alpha}(x)\partial_{x^\nu}^{\phantom{1}}\phi^{\phantom{1}}_\alpha(x)\Big)\hat{\chi}_{l,\alpha}^{\phantom{1}}\big(\hat{\varphi}_{l,\alpha}^{-1}(x)\big)
+\\+\frac{1}{2}
\sum_{\alpha\in \hat{I}_r}\int_{\hat{V}_{r,\alpha}}\mathrm{d}^{n-1}x\,g^{1/2}_\alpha(x)
\phi^{\phantom{1}}_\alpha(x)
\Big(\mathbf{n}^\nu_{r,\alpha}(x)\partial_{x^\nu}^{\phantom{1}}\phi^{\phantom{1}}_\alpha(x)\Big)\hat{\chi}_{r,\alpha}^{\phantom{1}}\big(\hat{\varphi}_{r,\alpha}^{-1}(x)\big),
\end{multline}
where the differential operator in each coordinate chart $U_\alpha$ is defined in local coordinates as follows
\begin{equation}\label{sk-3-1}
A_\alpha^{\phantom{1}}(x)=-g^{-1/2}_\alpha(x)
\partial_{x^\mu}g^{\mu\nu}_\alpha(x)g^{1/2}_\alpha(x)\partial_{x^\nu}+m^2.
\end{equation}
\begin{defa}\label{sk-d-10}
The symbol $A$ denotes a second-order differential operator, which in each individual coordinate chart $(U_{\alpha},\varphi_{\alpha})$ of the atlas $\mathcal{A}$ is defined by local formula \eqref{sk-3-1}, that is, $A_\alpha(x)=A(p)$, where $\varphi_{\alpha}(p)=x$. On a submanifold from the set $\{\mathcal{M}_l,\mathcal{M}_r,\mathcal{M}_{l,\Lambda},\mathcal{M}_{r,\Lambda}\}$ the operator $A$ is defined in a similar way, but using the corresponding atlas, see Section \ref{sk:sec:def:4}.
\end{defa}
\begin{defa}\label{sk-d-11}
Let $i\in\{l,r\}$. The symbol $N_{i}$ denotes the differential operator of the normal derivative acting from $C^\infty(\mathcal{M},\mathbb{R})$ to $C^\infty(Y_i,\mathbb{R})$, which in each separate coordinate chart $(\hat{U}_{i,\alpha},\hat{\varphi}_{i,\alpha})$, where $\alpha\in\hat{I}_i$, is defined by local formula $N_{i}^{\phantom{1}}(q)=\mathbf{n}^\nu_{l,\alpha}(x)\partial_{x^\nu}^{\phantom{1}}$, where $\hat{\varphi}_{i,\alpha}(q)=x$. In this case, the partial derivative itself is calculated near the point under consideration inside $\mathcal{M}$  with a further transition to the boundary. The restriction $N_{i}$ to the subset $C^\infty(\mathcal{M}_i,\mathbb{R})$ is notated by the same symbol.  Similarly, we define the normal derivative operator $N_{\Sigma,i}$, acting from $C^\infty(\mathcal{M}_i,\mathbb{R})$ to $C^\infty(\Sigma,\mathbb{R})$.
\end{defa}

Taking into account the latest definitions, formula \eqref{sk-3-2} can be written in a more concise form. Indeed, we introduce the measure of integration on the manifold $\mathcal{M}$ as follows
\begin{equation}\label{sk-3-3}
\mathrm{d}^np=\sum_{\alpha\in I}\mathrm{d}^nx\,g^{1/2}_\alpha\big(\varphi_\alpha^{\phantom{1}}(p)\big)
\chi_\alpha^{\phantom{1}}(p),
\end{equation}
where $x=\varphi_\alpha(p)\in V_\alpha$ in the corresponding coordinate chart. Introducing the notation in a similar way\footnote{The left part of \eqref{sk-3-3} does not contain an index corresponding to a manifold under consideration, since the integral contains an integration domain, and thus confusion does not arise.} for the submanifolds, we get
\begin{equation}\label{sk-3-4}
	S_0[\phi;\mathcal{M}]=\frac{1}{2}\int_{\mathcal{M}}\mathrm{d}^np\,\phi(p)A(p)\phi(p)+\frac{1}{2}
	\sum_{i\in\{l,r\}}\int_{Y_i}\mathrm{d}^{n-1}q\,\phi(q)N_i(q)\phi(q).
\end{equation}
In the same way, we receive
\begin{equation}\label{sk-3-5}
S_{\mathrm{int}}[\phi;\mathcal{M}]=
\sum_{k=3}^{+\infty}t_k\int_{\mathcal{M}}\mathrm{d}^np\,
\phi^k(p).
\end{equation}

Let us move on to the discussion of the second-order differential operator $A$. It is known from the general theory that the spectral problem on the manifold $\mathcal{M}$ with zero Dirichlet conditions 
\begin{equation}\label{sk-3-6}
\begin{cases}
A(p)\psi_\lambda(p)=\lambda\psi_\lambda(p),\,\,\,p\in\mathcal{M}\setminus\partial\mathcal{M};\\
\psi_\lambda(p)=0,\,\,\,p\in\partial\mathcal{M},
\end{cases}
\end{equation}
is well posed. The spectrum for this problem is countable, bounded from below by $\lambda\geqslant m^2$, see Section $8$ in \cite{sk-b-5}, and the only accumulation point is infinity. The corresponding Green's function
\begin{equation}\label{sk-3-8}
G(p_1,p_2)=\sum_\lambda\psi_\lambda(p_1)\frac{1}{\lambda}\psi_\lambda(p_2)
\end{equation}
for this task is well defined on the set $\mathcal{M}\times\mathcal{M}$ without the diagonal $\{p\times p:p\in\mathcal{M}\}$, satisfies the differential equation\footnote{In the general position on a Riemannian manifold, $\delta(p_1,p_2)$ denotes the kernel of the unit operator. In the particular Euclidean case, such a kernel is proportional to the standard $n$-dimensional delta function that depends on the difference of the arguments.}
\begin{equation}\label{sk-3-7}
A(p_1)G(p_1,p_2)=\sum_\lambda\psi_\lambda(p_1)\psi_\lambda(p_2)=\delta(p_1,p_2),
\end{equation}
as well as zero conditions on the boundary. Functions for the manifolds $\mathcal{M}_l$ and $\mathcal{M}_r$ are constructed in a similar way.
\begin{defa}\label{sk-d-12}
The symbols $G(\,\cdot\,,\,\cdot\,)$, $G_l(\,\cdot\,,\,\cdot\,)$, and $G_r(\,\cdot\,,\,\cdot\,)$ denote Green's functions for the operator $A$ on the manifolds $\mathcal{M}$, $\mathcal{M}_l$, and $\mathcal{M}_r$, respectively, with zero Dirichlet conditions on the boundary.
\end{defa}
\begin{lemma}\label{sk-l-4}
Let $\eta\in C^\infty(Y,\mathbb{R})$ such that $\eta=\eta_l\in C^\infty(Y_l,\mathbb{R})$ on $Y_l$ and $\eta=\eta_r\in C^\infty(Y_r,\mathbb{R})$ on $Y_r$. Then the solution $\phi^\eta(\cdot)$ to the problem
\begin{equation}\label{sk-3-9}
\begin{cases}
A(p)\phi^\eta(p)=0,\,\,\,p\in\mathcal{M}\setminus\partial\mathcal{M};\\
\phi^\eta(p)=\eta_i(p),\,\,\,p\in Y_i,\,\,\,i\in\{l,r\},
\end{cases}
\end{equation}
belongs to the class $C^\infty(\mathcal{M},\mathbb{R})$ and is represented as the following integral
\begin{equation}\label{sk-3-10}
\phi^{\eta}(p)=-\sum_{i\in\{l,r\}}\int_{Y_i}\mathrm{d}^{n-1}q\,\Big(N_i(q)G(p,q)\Big)\eta_i(q).
\end{equation}
\end{lemma}
\noindent This statement is well-known. Belonging to the class is proved nontrivially, while formula \eqref{sk-3-10} itself is obtained by integration by parts. The proof and detailed calculations can be found in Section $5.1$ of monograph \cite{sk-b-4}. Note that the operator that maps a function to a given behavior on the boundary has a single continuous extension to the mapping from the space\footnote{Regarding Sobolev spaces, see Section $4$ in \cite{sk-b-4}.} $H^{s}(Y)$ to $H^{s+1/2}(\mathcal{M})$ for $s\geqslant 1/2$, see Proposition $1.7$ in the same section\footnote{In fact, it can be expanded to $s\geqslant-1/2$, see Proposition $1.8$.}.
\begin{defa}\label{sk-d-13}
Let $i\in\{l,r\}$, and $\eta\in C^\infty(\partial\mathcal{M}_i,\mathbb{R})$ such that $\eta=\eta_i\in C^\infty(Y_i,\mathbb{R})$ on $Y_i$ and $\eta=\eta_\Sigma\in C^\infty(\Sigma,\mathbb{R})$ on $\Sigma$. The symbol $\phi^\eta_i(\cdot)$ denotes the solution to the problem \eqref{sk-3-9} on the submanifold $\mathcal{M}_i$, which, according to Lemma \ref{sk-l-4}, has the following explicit form
\begin{equation}\label{sk-3-11}
\phi^{\eta}_i(p)=-\int_{Y_i}\mathrm{d}^{n-1}q\,\Big(N_i(q)G_i(p,q)\Big)\eta_i(q)
-\int_{\Sigma}\mathrm{d}^{n-1}q\,\Big(N_{\Sigma,i}(q)G_i(p,q)\Big)\eta_\Sigma(q).
\end{equation}
\end{defa}
\begin{remark}\label{sk-r-3}
Under the conditions of Definition \ref{sk-d-13}, the functions $\eta_i$ and $\eta_\Sigma$ are given on disconnected smooth submanifolds of the same dimension. It is clear that any of the functions can be extended to another disconnected piece by zero, thus expanding the domain of definition. For convenience, we will further assume that the sums $\eta_i+\eta_\Sigma$ and $\eta_l+\eta_r$ are smooth functions on the sets $Y_i\cup\Sigma$ and $Y_l\cup Y_r$, respectively. At the same time, a similar reasoning can also be extended to the solutions $C^\infty(\mathcal{M},\mathbb{R})\ni\phi^{\eta_l+\eta_r}_{\phantom{i}}=\phi^{\eta_l}_{\phantom{i}}+\phi^{\eta_r}_{\phantom{i}}$ and $C^\infty(\mathcal{M}_i,\mathbb{R})\ni\phi^{\eta_i+\eta_\Sigma}_i=\phi^{\eta_i}_i+\phi^{\eta_\Sigma}_i$.
\end{remark}

\begin{lemma}\label{sk-l-5} Let $i\in\{l,r\}$, $\eta_i\in C^\infty(Y_i,\mathbb{R})$, and $\eta_\Sigma\in C^\infty(\Sigma,\mathbb{R})$. Next, introduce several functionals
\begin{equation}\label{sk-3-12}
S_{l,r}(\eta_l,\eta_r;\mathcal{M})=\int_{Y_l}\mathrm{d}^{n-1}q_1\int_{Y_r}\mathrm{d}^{n-1}q_2\,
\eta_l(q_1)
\Big(N_{l}(q_1)N_{r}(q_2)G(q_1,q_2)\Big)
\eta_r(q_2)=S_{r,l}(\eta_r,\eta_l;\mathcal{M}),
\end{equation}
\begin{equation}\label{sk-3-13}
S_{i,\Sigma}(\eta_i,\eta_\Sigma;\mathcal{M}_i)=\int_{Y_i}\mathrm{d}^{n-1}q_1\int_{\Sigma}\mathrm{d}^{n-1}q_2\,
\eta_i(q_1)
\Big(N_{i}(q_1)N_{\Sigma,i}(q_2)G_i(q_1,q_2)\Big)
\eta_\Sigma(q_2)=S_{\Sigma,i}(\eta_\Sigma,\eta_i;\mathcal{M}_i).
\end{equation}
%\begin{equation}\label{sk-3-14}
%S_{\Sigma,\Sigma}(\eta_i,\eta_\Sigma;\mathcal{M}_i)=\int_{Y_i}\mathrm{d}^{n-1}q_1\int_{\Sigma}\mathrm{d}^{n-1}q_2\,\eta_\Sigma(q_1)
%\Big(N_{\Sigma,i}(q_1)N_{\Sigma,i}(q_2)G_i(q_1,q_2)\Big)
%\eta_\Sigma(q_2),
%\end{equation}
Then for all $\phi\in C^\infty_0(\mathcal{M},\mathbb{R})$ and $\phi_i\in C^\infty_0(\mathcal{M}_i,\mathbb{R})$, where the subscript means that the functions vanish at the boundary of the corresponding manifold, the following relations are true
\begin{equation}\label{sk-3-15}
S_0[\phi+\phi^{\eta_l+\eta_r},\mathcal{M}]=S_0[\phi,\mathcal{M}]+S_0[\phi^{\eta_l},\mathcal{M}]+S_0[\phi^{\eta_r},\mathcal{M}]-S_{l,r}(\eta_l,\eta_r;\mathcal{M}),
\end{equation}
\begin{equation}\label{sk-3-16}
	S_0[\phi_i+\phi^{\eta_i+\eta_\Sigma},\mathcal{M}_i]=S_0[\phi_i,\mathcal{M}_i]+S_0[\phi^{\eta_i},\mathcal{M}_i]+S_0[\phi^{\eta_\Sigma},\mathcal{M}_i]-
	S_{i,\Sigma}(\eta_i,\eta_\Sigma;\mathcal{M}_i).
\end{equation}
\end{lemma}
\noindent The last statement can be found in \cite{sk-16}, see calculations after formula $(3.5)$. It follows after substituting the functions in \eqref{sk-2-1} and \eqref{sk-3-4} and performing integration by parts.
\begin{remark}\label{sk-r-4}
Note that all the terms in the right hand sides of \eqref{sk-3-15} and \eqref{sk-3-16} are well defined and finite, since the solution of \eqref{sk-3-10} is smooth, and the assumption from the Remark \ref{sk-r-2} is fulfilled. The last condition guarantees that the non-integrable behaviour on the diagonal is not present in the kernels from formulas \eqref{sk-3-12} and \eqref{sk-3-13}.
\end{remark}
\begin{defa}\label{sk-d-14}
Let $i\in\{l,r\}$. The symbol $D(Y_i,\mathcal{M})(\cdot)$ denote a pseudodifferential operator of the first order acting from $C^\infty(Y_i,\mathbb{R})$ to $C^\infty(Y_i,\mathbb{R})$ according to the rule
\begin{equation}\label{sk-3-19}
D(Y_i,\mathcal{M})(q)\eta_i(q)=N_{i}(q)\phi^{\eta_i}(q)
\end{equation}
for all $\eta_i\in C^\infty(Y_i,\mathbb{R})$ and $q\in Y_i$. Similarly we define $D(Y_i,\mathcal{M}_i)(\cdot)$ and $D(\Sigma,\mathcal{M}_i)(\cdot)$.
\end{defa}
\begin{col}\label{sk-c-5}
Let $i\in\{l,r\}$. Taking into account the results of Lemmas \ref{sk-l-4} and \ref{sk-l-5} and Definition \ref{sk-d-14}, the relations are correct
\begin{align}\label{sk-3-20}
S_0[\phi^{\eta_i},\mathcal{M}]&=\frac{1}{2}\int_{Y_i}\mathrm{d}^{n-1}q\,\eta_i(q)D(Y_i,\mathcal{M})(q)\eta_i(q)
,\\\label{sk-3-21}
S_0[\phi^{\eta_i},\mathcal{M}_i]&=\frac{1}{2}\int_{Y_i}\mathrm{d}^{n-1}q\,\eta_i(q)D(Y_i,\mathcal{M}_i)(q)\eta_i(q)
,\\\label{sk-3-22}
S_0[\phi^{\eta_\Sigma^{\phantom{1}}},\mathcal{M}_i]&=\frac{1}{2}\int_{\Sigma}\mathrm{d}^{n-1}q\,\eta_{\Sigma}(q)D(\Sigma,\mathcal{M}_i)(q)\eta_{\Sigma}(q).
\end{align}
\end{col}
\begin{remark}\label{sk-r-5}
Note that, taking into account Definition \ref{sk-d-14}, from a formal point of view, it is possible to write out integral kernels corresponding to the $D$-operators. However, such a representation is rough due to the fact that on the diagonal such kernels contain a non-integrable density\footnote{For example, look at formula $(3.18)$ from Section $5.3$ of monograph \cite{sk-b-4}. The normal derivative of the kernel at the boundary contains a non-integrable density.}. This leads to the need to talk about "unnecessary" regularization. For this reason, we preserve the quadratic shapes on the boundaries in a mathematically clean form, represented by the formulas from Corollary \ref{sk-c-5}.
\end{remark}
\begin{lemma}\label{sk-l-6}
The kernel of the integral operator, which is inverse one to the sum of the operators $D(\Sigma,\mathcal{M}_l)+D(\Sigma,\mathcal{M}_r)$, which acts from $C^\infty(\Sigma,\mathbb{R})$ to $C^\infty(\Sigma,\mathbb{R})$, is equal to the function $G(\,\cdot\,,\,\cdot\,)$, restricted on $\Sigma\times\Sigma$, that is,
\begin{equation}\label{sk-3-18}
\Big(D(\Sigma,\mathcal{M}_l)(q_1)+D(\Sigma,\mathcal{M}_r)(q_1)\Big)G(q_1,q_2)=\delta(q_1,q_2),
\end{equation}
where $\delta(q_1,q_2)$ is the kernel of the unit operator on the submanifold $\Sigma$.
\end{lemma}
\noindent The proof can be found in \cite{sk-b-6}. Additional information about the $D$-operator can be found in \cite{sk-b-7,sk-b-8,sk-b-9,sk-b-10}. Nevertheless, we note that the operator, taking into account the assumptions of this work, can be extended on $L^2$-space and, for $m^2>0$, has a positive spectrum (without zero eigenvalues), and the only point of accumulation is infinity.
\begin{lemma}\label{sk-l-7}
Let $i\in\{l,r\}$. Taking into account all the above, the difference $D(Y_i,\mathcal{M})(\cdot)-D(Y_i,\mathcal{M}_i)(\cdot)$ can be represented by an integral operator, the kernel of which, for each variable, is an integrable function on $Y_i$.
\end{lemma}
\begin{proof} First, note that the Green's functions $G(p_1,p_2)$ and $G_i(p_1,p_2)$ can be constructed using the reflection method, see \cite{sk-b-12} and the discussion in Section 2.3 of \cite{sk-16}. Thus, for the proof, it is sufficient to consider the standard decomposition near the diagonal for a Green's function on a compact smooth manifold without a boundary. The non-smooth component of such an asymptotic decomposition is local in nature\footnote{A summary of some known results related to the coefficients of local heat kernels can be found in \cite{sk-b-11,sk-b-13}.},  that is, it is determined by the coefficients of the Laplace operator, restricted to a small neighborhood containing the points $p_1,p_2$ and the geodesic connecting them, see \cite{15} and formulas (107) and (108) in \cite{sk-b-11}. It is constructed using the Seeley--DeWitt coefficients \cite{10-ss,111,8,110}, the Synge's world function \cite{104,1040}, as well as the Van Vleck--Morett determinant  \cite{105-ss}. In particular, the main order near the diagonal is proportional to $\ln(d(p_1,p_2))$ for $n=2$ and $\big(d(p_1,p_2)\big)^{2-n}$ for $n>2$.

Next, let $\eta_i\in C^\infty(Y_i,\mathbb{R})$. Consider the difference $N_{i}(q_1)\phi^{\eta_i}_{\phantom{i}}(q_1)-N_{i}(q_1)\phi^{\eta_i}_i(q_1)$ and use Lemma \ref{sk-l-4} and Definition \ref{sk-d-13}, then the kernel for the difference of the operators can be written as
\begin{equation}\label{sk-3-23}
-N_i(q_1)N_i(q_2)G(q_1,q_2)+N_i(q_1)N_i(q_2)G_i(q_1,q_2).
\end{equation}
Therefore, applying two derivatives to the Green's functions and moving to the boundary $Y_i$, we make sure that the non-smooth terms of the asymptotics are reduced, and the remaining part is an integrable function\footnote{Note that this proof based on the local behaviour can be replaced by differentiating the gluing  formula from \eqref{sk-6-3} for the Green's functions, taking into account Remark \ref{sk-r-2}.} in the $(n-1)$-dimensional space $Y_i$.
\end{proof}

\subsection{Effective action}
\label{sk:sec:def:3}

In this section, a special type of transformation is formulated, given on a set of classical actions corresponding \footnote{Within the main part of the text.} to the description of scalar theories. As a rule, such a transformation is called functional (or path) integration. Nevertheless, it is worth noting right away that at the moment no mathematical formalism has been built, so we have to work with a formal series of a special kind, the internal structure of which satisfies a number of properties of the standard integral, for example, integration by parts and changes of variables. Some discussions of such "integration" can be found in works \cite{sk-1,sk-2,3,sk-b-14}.

Let $\eta\in C^\infty(Y,\mathbb{R})$ such that $\eta=\eta_l\in C^\infty(Y_l,\mathbb{R})$ on $Y_l$ and $\eta=\eta_r\in C^\infty(Y_r,\mathbb{R})$ on $Y_r$. Let us consider the formal notation of the functional integral
\begin{equation}\label{sk-5-1}
e^{-W_{\mathrm{eff}}[\sqrt{\hbar}\eta;\mathcal{M}]/\hbar}=
\mathcal{N}^{-1}(\mathcal{M})
\int\limits_{\mathcal{H}(\sqrt{\hbar}\eta;\mathcal{M})}
\mathcal{D}\phi\,e^{-S_{\mathrm{cl}}[\phi;\mathcal{M}]/\hbar}
\end{equation}
and comment on its individual components. In the mentioned formula a set of functions $\mathcal{H}\big(\sqrt{\hbar}\eta;\mathcal{M}\big)$ is called the integration domain and contains functions, which on the boundary $\partial\mathcal{M}$ are equal to $\sqrt{\hbar}\eta$. It is impossible to add something more specific about the smoothness properties of such functions due to the lack of a general theory. The multiplier $\mathcal{N}(\mathcal{M})$ denotes normalization based on the assumption that in the absence of interaction and zero boundary conditions, equality is satisfied 
\begin{equation}\label{sk-5-3}
W_{\mathrm{eff}}[0;\mathcal{M}]\Big|_{\{t_k=0\}_{k=3}^{+\infty}}=0.
\end{equation}
In other words, with zero external data, such an object does not carry important information from a physical point of view. 
Further, the symbol $\mathcal{D}\phi$ denotes the "measure" of integration, which, generally speaking, depends on the manifold $\mathcal{M}$. It is assumed that linear change of variables can be done in the integral, so formula \eqref{sk-5-1} can be reduced to some "basic" integral
\begin{equation}\label{sk-5-2}
e^{-W_{\mathrm{eff}}[\sqrt{\hbar}\eta;\mathcal{M}]/\hbar}=
\mathcal{N}^{-1}(\mathcal{M})
\int\limits_{\mathcal{H}(0;\mathcal{M})}
\mathcal{D}\phi\,e^{-S_{\mathrm{cl}}[\sqrt{\hbar}\phi+\sqrt{\hbar}\phi^\eta;\mathcal{M}]/\hbar},
\end{equation}
where we made the change of variables of the form $\phi\to\sqrt{\hbar}\phi^\eta+\sqrt{\hbar}\phi$. This substitution leads to the use of the background field method, see \cite{12,102,103,24,25,26,23,Ivanov-Russkikh}, with the only difference that the field  $\sqrt{\hbar}\phi^\eta$, responsible for the shift, does not solve the quantum equation of motion, but is a solution to the classical problem from Lemma \ref{sk-l-4}.

Let us introduce two additional auxiliary objects into consideration
\begin{multline}\label{sk-5-6}
S(\phi+\phi^\eta;\mathcal{M})=
\exp\Big(-S_{\mathrm{int}}[\sqrt{\hbar}\phi+\sqrt{\hbar}\phi^\eta;\mathcal{M}]/\hbar\Big)=\\=
1+\sum_{k\geqslant3}\hbar^{k/2-1}
\Bigg(\prod_{i=1}^{k}\int_{\mathcal{M}}\mathrm{d}^np_{i}\Bigg)
S_{k}(p_1,\ldots,p_k)
\big(\phi+\phi^\eta\big)(p_1)\ldots\big(\phi+\phi^\eta\big)(p_{k}),
\end{multline}
and also
\begin{equation}\label{sk-5-5}
g(\psi;\mathcal{M})=\frac{1}{2}\int_{\mathcal{M}}\mathrm{d}^np_1\int_{\mathcal{M}}\mathrm{d}^np_1
\psi(p_1)G(p_1,p_2)\psi(p_2),
\end{equation}
where $\psi\in C^\infty_0(\mathcal{M},\mathbb{R})$.
\begin{lemma}\label{sk-l-8}
The density $S_{k}(p_1,\ldots,p_k)$ from  formula \eqref{sk-5-6} for each fixed index value contains a finite number of terms, each of which is represented as a product of delta-functionals and constants. In this case, the arguments of singular functionals are only the points $\{p_i\}_{i=1}^k$.
\end{lemma}
\begin{proof}
It follows from the construction of the functional, see \eqref{sk-5-6}.
\end{proof}
\begin{remark}\label{sk-r-7}
In the density $S_{k}(p_1,\ldots,p_k)$, delta-functionals do not appear with matching arguments (on the diagonal), since there are no divergences in the left part of \eqref{sk-5-6}.
\end{remark}
Taking into account the latest definitions, formula \eqref{sk-5-2} can be rewritten in an equivalent way
\begin{equation}\label{sk-5-8}
e^{-W_{\mathrm{eff}}[\sqrt{\hbar}\eta;\mathcal{M}]/\hbar}=e^{-S_{\mathrm{0}}[\phi^\eta;\mathcal{M}]}
\Bigg(\mathcal{N}^{-1}(\mathcal{M})
\int\limits_{\mathcal{H}(0;\mathcal{M})}
\mathcal{D}\phi\,e^{-S_{0}[\phi;\mathcal{M}]}S(\phi+\phi^\eta;\mathcal{M})
\Bigg).
\end{equation}
\begin{defa}\label{sk-d-16}
The functional integral \eqref{sk-5-8} should be understood as an formal series with respect to the parameter $\sqrt{\hbar}$ of the form
\begin{equation}\label{sk-5-7}
e^{-W_{\mathrm{eff}}[\sqrt{\hbar}\eta;\mathcal{M}]/\hbar}=e^{-S_{\mathrm{0}}[\phi^\eta;\mathcal{M}]}
\bigg(
S(\delta_\psi+\phi^\eta;\mathcal{M})e^{g(\psi;\mathcal{M})}\bigg)\bigg|_{\psi=0},
\end{equation}
where the functional derivative $\delta_\psi$ at the point $p_1$ is defined by equality $\delta_{\psi(p_1)}\psi(p_2)=\delta(p_1,p_2)$.
\end{defa}
\begin{proof}
The formula is obtained after taking out the $S$-functional in the form of functional derivatives acting on a linear additive from under the integral and calculating Gaussian integrals with the quadratic form $S_{0}[\phi;\mathcal{M}]$ and the linear additive, see \cite{sk-29-3}.
\end{proof}
\begin{remark}\label{sk-r-6}
The representation from \eqref{sk-5-7} is completely equivalent to writing through a sum of Feynman diagrams. In this paper, an alternative entry is chosen to save space and convenience, so as not to formulate the rules of diagrammatic technique.
\end{remark}
\begin{defa}\label{sk-d-15}
The functional $W_{\mathrm{eff}}[\eta;\mathcal{M}]$ is called a quantum (or effective) action corresponding to the classical action $S_{\mathrm{cl}}[\,\cdot\,;\mathcal{M}]$, which is considered on the set of functions that take the value $\eta\in C^\infty(Y,\mathbb{R})$ on the boundary.
\end{defa}
\begin{remark}\label{sk-r-8}
All definitions and arguments in this section easily extend to the submanifolds $\mathcal{M}_l$ and $\mathcal{M}_r$. To do this, it is enough to take into account changes of manifolds, boundary conditions and the corresponding Green's functions. Analogs of the functional from \eqref{sk-5-5} on the submanifolds will be notated as $g_l(\psi_l;\mathcal{M}_l)$ and $g_r(\psi_r;\mathcal{M}_r)$. They are constructed using the Green's functions $G_l$ and $G_r$ on the corresponding submanifolds.
\end{remark}

As it is knowm, when calculating the functional in parentheses from \eqref{sk-5-7}, the Wick's theorem is used, according to which all possible pairs of fields $\phi(p_i)\phi(p_j)$ from decomposition \eqref{sk-5-6} are replaced by the Green's functions $G(p_i,p_j)$. Thus, nonlinear combinations of the Green's functions appear under the sign of integration. Given the fact that the Green function has a singularity on the diagonal (for $p_i=p_j$), nonlinear combinations can lead to two types of divergences. In one case, the appearance of a non-integrable density is possible, and in the second case, the presence of a singular functional at a singular point.

As an example, we can consider a quartic model in three-dimensional space. It is clear that the combination of the Green's functions
\begin{equation}\label{sk-5-9}
\Big(G(p_i,p_j)\Big)^3\sim\frac{1}{\big(d(p_i,p_j)\big)^3}
\end{equation}
for $p_i\sim p_j$, where $p_i,p_j\in\mathcal{M}\setminus Y$, and therefore it is a non-integrable function. At the same time, there is a singular point on the diagonal $G(p_i,p_j)$ for $p_i=p_j$.

It turns out that both types of divergences can be regularized using a special "cutoff" in the coordinate representation. Moreover, such regularization is consistent with the process of gluing of manifolds and partitial functions, and the exact definition is given in Section  \ref{sk:sec:reg:3}. It is the description of the regularization process and the proof of consistency that the rest of the work will be devoted to.

\section{Cutoff regularization}
\label{sk:sec:reg}

This section is devoted to the formulation of a generalization of a cutoff regularization in the coordinate representation, which was previously successfully applied in the standard Euclidean space $\mathbb{R}^n$ with Cartesian coordinates $\{x_i\}_{i=1}^n$,  in the case of a smooth $n$-dimensional Riemannian manifold $\mathcal{M}$ with a boundary. Further we assume that $n$ is fixed.

\subsection{Deformation in the Euclidean case}
\label{sk:sec:reg:1}

Let us start with some useful facts about the mentioned regularization \cite{34,Iv-2024-1,Ivanov-Akac,Ivanov-Kharuk-2020,Ivanov-Kharuk-20222,Ivanov-Kharuk-2023,Iv-Kh-2024,Kh-2024,Ivanov-2022,Iv-2024}. To do this, we introduce a fundamental solution\footnote{The paper does not consider the case of $n=1$. Some comments on the introduction of a regularization for $n=1$ can be found in article \cite{Ivanov-2022}.}$^{,}$\footnote{Note that the Laplace operator and the corresponding fundamental solution depend on the dimension. Assuming $n$ is fixed, we omit this index within this section.} for the following free\footnote{That is, without any potentials.} Laplace operator 
\begin{equation}
A(x)=-\sum_{i=1}^n\partial_{x^i}^2
\end{equation}
in the form
\begin{equation}\label{sk-1-3}
G(x)=\frac{|x|^{2-n}}{(n-2)S_{n-1}},\,\,\,
S_{n-1}=\frac{2\pi^{n/2}}{\Gamma(n/2)}
\end{equation}
for the case $n>2$, and $G(x)=-\ln(|x|)/(2\pi)$ for $n=2$. Here $S_{n-1}$ is the area of a $(n-1)$-dimensional unit sphere. It is known that in the sense of generalized functions on the Schwartz class, see \cite{Gelfand-1964,Vladimirov-2002}, the equality $A(x)G(x-y)=\delta(x-y)$ is valid. Further, let $k\in\mathbb{N}$, and $\alpha=(\alpha_1,\ldots,\alpha_k)$ is a multi-index\footnote{Here the designation is used, in which $0\notin\mathbb{N}$.}, the components of which satisfy the following relations
\begin{equation}\label{sk-1-1}
\alpha_i\in(0,1]\,\,\,\mbox{for all}\,\,\,i\in\{1,\ldots,k\},\,\,\,\mbox{and also}\,\,\,
\sum_{i=1}^k\alpha_i\leqslant1.
\end{equation}

\begin{defa}\label{sk-d-2}
Let us consider an integral operator\footnote{Note that such an operator is defined on a wider class of functions. However, within the framework of this section, the specified set is quite sufficient to determine.} $\mathrm{H}_\alpha^\Lambda:C^\infty(\mathbb{R}^n,\mathbb{R})\to C^\infty(\mathbb{R}^n,\mathbb{R})$, which acts on the function $\phi\in C^\infty(\mathbb{R}^n,\mathbb{R})$ according to the following relation
\begin{equation}\label{sk-1-2}
\mathrm{H}_\alpha^\Lambda(\phi)(x)=
\int_{\mathrm{S}^{n-1}}\frac{\mathrm{d}^{n-1}\sigma(\hat{x}_k)}{S_{n-1}}\ldots
\int_{\mathrm{S}^{n-1}}\frac{\mathrm{d}^{n-1}\sigma(\hat{x}_1)}{S_{n-1}}\,
\phi\Bigg(x+\Lambda^{-1}\sum_{i=1}^k\hat{x}_i\alpha_i\Bigg),
\end{equation}
where $\hat{x}_i\in\mathrm{S}^{n-1}\subset\mathbb{R}^n$ is an element of a unit sphere centered at the origin, and integration is performed using the standard measure on the sphere. 
\end{defa}

\noindent It can be seen from the formula that the operator is a multiple averaging over spheres near the point $x\in\mathbb{R}^n$, while only the elements of the ball are involved in the averaging process $\{y\in\mathbb{R}^n:|x-y|\leqslant1/\Lambda\}$. Also, in the limit $\Lambda\to+\infty$, the function $\mathrm{H}_\alpha^\Lambda(\phi)$ tends to $\phi$, because the support of the kernel of the integral operator tends to zero.

\begin{defa}\label{sk-d-1}
The integral transformation $\mathrm{H}_\alpha^\Lambda(\phi)$ of a function $\phi$ is called its deformation, while the mapping $\mathrm{H}_\alpha^\Lambda$ is named a deforming operator.
\end{defa}

\noindent Note that the function $G(x)$ is integrable in $(n-1)$-dimensional space, therefore, the deforming operator $\mathrm{H}_\alpha^\Lambda$ can be applied to it. A function $\mathrm{H}_\alpha^\Lambda(G)$ for $n>2$ was investigated in the work \cite{Iv-2024}, and for $n=2$ was studied in \cite{sk-b-15}. Let us formulate some useful properties for further analysis in the form of lemma.

\begin{lemma}\label{sk-l-1}
Taking into account all of the above, the following relation holds
\begin{equation}\label{sk-1-4}
\mathrm{H}_\alpha^\Lambda(G)(x)=
\begin{cases} 
\Lambda^{n-2}\mathbf{f}\big(|x|^2\Lambda^2\big)+G(x/(|x|\Lambda)), &|x|\leqslant1/\Lambda;\\
~~~~~~~~~~~~~~~~~G(x), &|x|>1/\Lambda,
\end{cases}
\end{equation}
where $\mathbf{f}(\cdot)\in C([0,1],\mathbb{R})$. The latter function depends on the dimension parameter $n$ and the multi-index $\alpha$, and also has the property $\mathbf{f}(1)=0$.
\end{lemma}

\begin{col}\label{sk-c-1}
The deformation $\mathrm{H}_\alpha^\Lambda(G)$ of the fundamental solution $G$ is a bounded function in any bounded domain of $\mathbb{R}^n$ for $n=2$ and is bounded in the entire space for $n>2$.
\end{col}

Before discussing a more general Riemannian manifold, we make an additional observation. Note that the sphere integration from \eqref{sk-1-2} can be rewritten in an equivalent way. Let $\phi\in C^\infty(\mathbb{R}^n,\mathbb{R})$, then
\begin{align*}
\int_{\mathrm{S}^{n-1}}\frac{\mathrm{d}^{n-1}\sigma(\hat{x}_1)}{S_{n-1}}\,\phi(x+\hat{x}_1\alpha_1/\Lambda)=&
\int_{\mathrm{S}^{n-1}}\frac{\mathrm{d}^{n-1}\sigma(\hat{x}_1)}{S_{n-1}}
\int_{0}^{+\infty}\mathrm{d}r\,r^{n-1}\delta(r-1)\phi(x+r\hat{x}_1\alpha_1/\Lambda)\\
=&\int_{\mathbb{R}^n}\frac{\mathrm{d}^{n}x_1}{S_{n-1}(\alpha_1/\Lambda)^n}\,\delta(|x_1|\Lambda/\alpha_1-1)\phi(x+x_1)\\=&
\int_{\mathbb{R}^n}\mathrm{d}^{n}x_1\bigg(\frac{\delta(|x-x_1|\Lambda/\alpha_1-1)}{S_{n-1}(\alpha_1/\Lambda)^n}\bigg)\phi(x_1)\\
=&\int_{\mathbb{R}^n}\mathrm{d}^{n}x_1\,\frac{\omega_0(|x-x_1|\Lambda/\alpha_1)}{(\alpha_1/\Lambda)^n}\phi(x_1),
\end{align*}
where the last equality defines a generalized function $\omega_0(\cdot)$. It is clear that $\omega_0(|\cdot\,-\,\cdot|\Lambda/\alpha_1)/(\alpha_1/\Lambda)^n$ is the kernel of the integral operator from $C^\infty(\mathbb{R}^n,\mathbb{R})$ to $C^\infty(\mathbb{R}^n,\mathbb{R})$. The latter representation makes it possible to generalize the deforming operator in a natural way while preserving the result of Lemma \ref{sk-l-1}.

\begin{defa}\label{sk-d-3}
Let $j\in\mathbb{N}\cup\{0\}$. A kernel $\omega(|\cdot\,-\,\cdot|)$ of an integral operator from $C^\infty(\mathbb{R}^n,\mathbb{R})$ to $C^\infty(\mathbb{R}^n,\mathbb{R})$ is called $j$-acceptable if the following relation holds
\begin{equation}\label{sk-1-7}
\int_{\mathbb{R}^n}\mathrm{d}^{n}x_1\,\omega(|x-x_1|)G(x_1)\in C^j(\mathbb{R}^n,\mathbb{R}).
\end{equation}
At the same time, in the designation of 0-acceptable, the corresponding index will usually be omitted.
\end{defa}

\begin{lemma}\label{sk-l-2}
Let the assumptions of Definition \ref{sk-d-2} be fulfilled, and $\tilde{\omega}=\{\omega_i(|\cdot\,-\,\cdot|\Lambda/\alpha_i)/(\alpha_i/\Lambda)^n\}_{i=1}^k$ is a set of acceptable kernels of integral operators from $C^\infty(\mathbb{R}^n,\mathbb{R})$ to $C^\infty(\mathbb{R}^n,\mathbb{R})$, satisfying the relations
\begin{equation}\label{sk-1-6}
\mathop{\mathrm{supp}}_{\mathbb{R}^n}(\omega_i)\subset\{y\in\mathbb{R}^n:|y|\leqslant1\}\,\,\,\mbox{and}\,\,\,
\int_{\mathbb{R}^n}\mathrm{d}^{n}x\,\omega_i(|x|)=S_{n-1}\int_{0}^{+\infty}\mathrm{d}r\,r^{n-1}\omega_i(r)=1
\end{equation}
for all $i\in\{1,\ldots,k\}$. Then for a deformed function of the form
\begin{equation}\label{sk-1-5}
\mathrm{H}_{\tilde{\omega},\alpha}^\Lambda(G)(x)=
\int_{\mathbb{R}^n}\mathrm{d}^{n}x_k\,\frac{\omega_k(|x-x_k|\Lambda/\alpha_k)}{(\alpha_k/\Lambda)^n}\ldots
\int_{\mathbb{R}^n}\mathrm{d}^{n}x_1\,\frac{\omega_1(|x_2-x_1|\Lambda/\alpha_1)}{(\alpha_1/\Lambda)^n}\,
G(x_1),
\end{equation}
the result of Lemma \ref{sk-l-1} is correct.
\end{lemma}
\begin{proof} Continuity of the function $\mathrm{H}_{\tilde{\omega},\alpha}^\Lambda(G)(\cdot)$ follows from the acceptability of the kernels, so it is enough to check only the fact that $\mathrm{H}_{\tilde{\omega},\alpha}^\Lambda(G)(x)=G(x)$ for all $|x|\geqslant 1/\Lambda$. To do this, it is enough to show that 
\begin{equation}\label{sk-1-8}
\int_{\mathbb{R}^n}\mathrm{d}^{n}x_i\,\frac{\omega_i(|x-x_i|\Lambda/\alpha_i)}{(\alpha_i/\Lambda)^n}\,
G(x_i)=G(x)
\end{equation}
for all $i$ and $|x|\geqslant\alpha_i/\Lambda$. Let us change the variable $x_i\to\alpha_ix_i/\Lambda+x$ in the last integral and use the result of Lemma \ref{sk-l-1} in the form of the relation
\begin{equation}\label{sk-1-10}
\int_{\mathrm{S}^{n-1}}\mathrm{d}^{n-1}\sigma(\hat{x}_i)
G(x+\alpha_ir\hat{x}_i/\Lambda)=S_{n-1}G(x)
\end{equation}
for all $|x|\geqslant\alpha_ir/\Lambda$, where $r>0$. Then, under the specified assumptions about the support of the kernel, see \eqref{sk-1-6}, we obtain
\begin{align*}
\int_{\mathbb{R}^n}\mathrm{d}^{n}x_i\,\omega_i(|x_i|)\,
G(x+\alpha_ix_i/\Lambda)=&
\int_{0}^{1}\mathrm{d}r\,r^{n-1}\omega_i(r)
\int_{\mathrm{S}^{n-1}}\mathrm{d}^{n-1}\sigma(\hat{x}_i)
G(x+\alpha_ir\hat{x}_i/\Lambda)\\=&\,
\bigg(S_{n-1}\int_{0}^{+\infty}\mathrm{d}r\,r^{n-1}\omega_i(r)\bigg)
G(x)=\,G(x)
\end{align*}
for all $|x|\geqslant\alpha_i/\Lambda$, from which the result of the lemma follows.
\end{proof}

\begin{lemma}\label{sk-l-3}
Let the assumptions of Lemma \ref{sk-l-2} be fulfilled, then the kernel of the integral operator
\begin{equation}\label{sk-1-11}
\int_{\mathbb{R}^n}\mathrm{d}^{n}x_k\,\frac{\omega_k(|x-x_k|\Lambda/\alpha_k)}{(\alpha_k/\Lambda)^n}\ldots
\int_{\mathbb{R}^n}\mathrm{d}^{n}x_2\,\frac{\omega_1(|x_3-x_2|\Lambda/\alpha_2)}{(\alpha_2/\Lambda)^n}
\frac{\omega_1(|x_2-x_1|\Lambda/\alpha_1)}{(\alpha_1/\Lambda)^n}
\end{equation}
is an acceptable kernel $\Lambda^n\omega(|x-x_1|\Lambda)$ of an integral operator from $C^\infty(\mathbb{R}^n,\mathbb{R})$ to $C^\infty(\mathbb{R}^n,\mathbb{R})$ and it satisfies the following relations
\begin{equation}\label{sk-1-12}
\mathop{\mathrm{supp}}_{\mathbb{R}^n}(\omega)\subset\{y\in\mathbb{R}^n:|y|\leqslant1\}\,\,\,\mbox{and}\,\,\,
\int_{\mathbb{R}^n}\mathrm{d}^{n}x\,\omega(|x|)=1.
\end{equation}
\end{lemma}
\begin{proof} For $k=1$, the case is trivial, so let $k\geqslant2$. Note that the combination under consideration depends on $|x-x_1|$. This is a consequence of the series of transformations $x_i\to x_i+x_1$ for $i\in\{2,\ldots,k\}$ and the spherical symmetry of the kernels. Next, $\omega(|x-x_1|\Lambda)$ is acceptable due to the acceptability of all components, and the restriction on the support follows from the restrictions on the supports of individual components. To check the latter relation, consider the corresponding integral and perform an additional series of changes $x_i\to\alpha_ix_i/\Lambda$ for $i\in\{2,\ldots,k\}$
\begin{equation}
\int_{\mathbb{R}^n}\mathrm{d}^{n}x\,\omega(|x|)=
\int_{\mathbb{R}^n}\mathrm{d}^{n}x
\int_{\mathbb{R}^n}\mathrm{d}^{n}x_k\,\omega_k(|x-x_k|)
\int_{\mathbb{R}^n}\mathrm{d}^{n}x_2\,\omega_2(|x_3-x_2|)
\omega_1(|x_2|).
\end{equation}
Then the statement follows after substituting $x\to x+x_k$ and $x_i\to x_i+x_{i-1}$ for $i\in\{3,\ldots,k\}$, leading to the factorization, and the usage of the normalizations from \eqref{sk-1-6}.
\end{proof}
\begin{defa}\label{sk-d-4}
The symbol $\Omega_j(\mathbb{R}^n)$ denotes a set of all $j$-acceptable kernels $\omega(|\cdot\,-\,\cdot|)$ of integral operators from $C^\infty(\mathbb{R}^n,\mathbb{R})$ to $C^\infty(\mathbb{R}^n,\mathbb{R})$, satisfying the constraints on the support and normalization from \eqref{sk-1-12}.
\end{defa}
\begin{col}\label{sk-c-2}
Examples of kernels from Definition \ref{sk-d-2} and Lemma \ref{sk-l-2} belong to $\Omega_0(\mathbb{R}^n)$.
\end{col}
\begin{defa}\label{sk-d-5}
From now on, the deforming operator is an integral operator $\mathrm{H}_\omega^\Lambda$ from $C^\infty(\mathbb{R}^n,\mathbb{R})$ to $C^\infty(\mathbb{R}^n,\mathbb{R})$ with a kernel $\Lambda^n\omega(|x-x_1|\Lambda)\in\Omega_0(\mathbb{R}^n)$, where $\Lambda>0$. A function
$\phi\in C^\infty(\mathbb{R}^n,\mathbb{R})$, 
transformed with the usage of this operator,
of the form
\begin{equation}\label{sk-1-13}
\phi_\omega^\Lambda(x)\equiv
\mathrm{H}_\omega^\Lambda(\phi)(x)=\int_{\mathbb{R}^n}\mathrm{d}^nx_1\,\Lambda^n\omega(|x-x_1|\Lambda)\phi(x_1)
\end{equation}
is called the deformation of the function $\phi$ with the kernel $\omega$ and the regularizing parameter $\Lambda$.
\end{defa}

\subsection{Generalization of the deformation}
\label{sk:sec:reg:2}

Let all the assumptions from Section \ref{sk:sec:def} be fulfilled. Let us extend the definition of the averaging operator to the case of a smooth Riemannian manifold. Note that at short distances, the main approximation of a Green's function for the Dirichlet problem with zero boundary conditions has the form of the fundamental solution from \eqref{sk-1-3} for the Euclidean case. However, unlike $\mathbb{R}^n$, the manifold $\mathcal{M}$ may have a boundary, so we should specify separately how the averaging occurs near the boundary. As before, the function $d(\,\cdot\,,\,\cdot\,)$, defined on $\mathcal{M}\times\mathcal{M}$, denotes the geodesic distance.

Before proceeding to the formulation of definitions, we note that kernels of integral operators in the case of a smooth manifold with a boundary in general position may depend on the points of the manifold itself (individually), and not only on the distance between the points, as it was in the case of flat metric. The dependence of a kernel from Definition \ref{sk-d-4} follows from the invariance with respect to shifts and rotations, which may be absent in the case of a general position.

\begin{defa}\label{sk-d-20}
Let $\Lambda>0$ and $p\in\mathcal{M}$, then $\mathcal{B}_{1/\Lambda}(p,\mathcal{M})=\{p_1\in\mathcal{M}:d(p,p_1)\leqslant1/\Lambda\}$.
\end{defa}
\begin{col}\label{sk-c-6}
If $Y\neq\emptyset$ and there exists a $q\in Y$ such that $d(p,q)<1/\Lambda$, then the set $\mathcal{B}_{1/\Lambda}(p,\mathcal{M})$ is a truncated ball containing part of the border $Y$.
\end{col}
\begin{remark}\label{sk-r-10}
	Further, we will assume that the parameter $\Lambda$ is so large that for any $p\in\mathcal{M}$ and any $p_1\in\mathcal{B}_{1/\Lambda}(p,\mathcal{M})$ there is a unique geodesic contained in $\mathcal{B}_{1/\Lambda}(p,\mathcal{M})$.
\end{remark}
\begin{defa}\label{sk-d-18}
Let $j\in\mathbb{N}\cup\{0\}$ and $\Lambda>0$. A kernel $\omega(\,\cdot\,,\,\cdot\,;\Lambda)$ of an integral operator from $C^\infty(\mathcal{M},\mathbb{R})$ to $C^j(\mathcal{M},\mathbb{R})$ is called $j$-acceptable if the relation holds
\begin{equation}\label{sk-7-1}
\int_{\mathcal{M}}\mathrm{d}^{n}p_1\,\omega(\,\cdot\,,p_1;\Lambda)G(p_1,p_2)\in C^j(\mathcal{M}_\Lambda,\mathbb{R})\cap
C(\mathcal{M},\mathbb{R})
\end{equation}
for all $p_2\in\mathcal{M}$, and limit transitions
\begin{equation}\label{sk-7-1-1}
A(p)\int_{\mathcal{M}}\mathrm{d}^{n}p_1\,\omega(p,p_1;\Lambda)G(p_1,p_2)\to\delta(p,p_2)
\end{equation}
and
\begin{equation}\label{sk-7-1-2}
\omega(p,p_1;\Lambda)\to\delta(p,p_1)
\end{equation}
are also valid for $\Lambda\to+\infty$ on the class $C^\infty(\mathcal{M},\mathbb{R})$.
At the same time, the index will be omitted in the name of $0$-acceptable ones. A set of $j$-acceptable kernels additionally satisfying the relations
%\begin{equation}\label{sk-7-4}
%\omega(p,p_1;\Lambda)=\omega(p,p_2;\Lambda)\,\,\,\mbox{для всех}\,\,\,
%p_1,p_2\in\{p_3\in\mathcal{B}_{1/\Lambda}(p,\mathcal{M}):d(p,p_3)=s\},\,\,\,\mbox{где}\,\,\, s\in[0,1/\Lambda],
%\end{equation}
\begin{equation}\label{sk-7-3}
\mathop{\mathrm{supp}}_{\mathcal{M}}(\omega(p,\,\cdot\,;\Lambda))\subset\mathcal{B}_{1/\Lambda}(p,\mathcal{M})
\,\,\,\mbox{for all}\,\,\,p\in\mathcal{M},\,\,\,\mbox{and also}\,\,\,
\int_{\mathcal{M}}\mathrm{d}^{n}p_1\,\omega(p,p_1;\Lambda)=1,
\end{equation}
is notated as $\Omega_{j}^\Lambda(\mathcal{M})$.
\end{defa}
\begin{remark}\label{sk-r-9}
In the integrals from \eqref{sk-7-1} and \eqref{sk-7-3} we actually integrate over the subset $\mathcal{B}_{1/\Lambda}(p,\mathcal{M})\subset\mathcal{M}$. The integrals themselves should be understood in the sense of partitioning into subsets using the corresponding coordinate charts from Sections \ref{sk:sec:def:1} and \ref{sk:sec:def:4}.
\end{remark}
%\begin{remark}\label{sk-r-11}
%Соотношение \eqref{sk-7-4} показывает, что ядро инвариантно относительно поворотов, сохраняющих длину геодезической, вокруг первого аргумента. Можно сказать, что $\omega$ является функцией точки $p$, расстояния $d(p,p_1)$ и регуляризующего параметра $\Lambda$. При этом, в отличие от случая евклидова пространства из определения \ref{sk-d-4}, трансляционаая инвариантность не подразумевается, так как присутствует нетривиальная метрика. 
%\end{remark}
\begin{defa}\label{sk-d-19}
A deforming operator on the Riemannian manifold $\mathcal{M}$ is the integral operator $\mathrm{H}_\omega^\Lambda$, which acts from $C^\infty(\mathcal{M},\mathbb{R})$ to $C^j(\mathcal{M},\mathbb{R})$, with a kernel $\omega(\,\cdot\,,\,\cdot\,;\Lambda)\in\Omega_j^\Lambda(\mathcal{M})$, where the parameter $\Lambda$ satisfies Remark \ref{sk-r-10}, as well as $\Lambda>\Lambda_1$ taking into account Remark \ref{sk-r-2}. A function $\phi\in C^\infty(\mathcal{M},\mathbb{R})$ transformed with the usage of this operator
	\begin{equation}\label{sk-7-2}
		\phi_\omega^\Lambda(p)\equiv
		\mathrm{H}_\omega^\Lambda(\phi)(p)=\int_{\mathcal{M}}\mathrm{d}^{n}p_1\,\omega(p,p_1;\Lambda)
		\phi(p_1)
	\end{equation}
	is called a deformation of the function $\phi$ with the kernel $\omega$ and the regularizing parameter $\Lambda$.
\end{defa}

As an example, we consider the kernel of averaging over a sphere by analogy with the operator from \eqref{sk-1-2} for $k=1$, which was previously studied for the case of the Euclidean space $\mathbb{R}^n$. Indeed, then relation \eqref{sk-7-1} is rewritten as
\begin{equation}\label{sk-7-5}
\int_{\mathcal{M}}\mathrm{d}^{n}p_1\,\omega_0(p,p_1;\Lambda)
\phi(p_1)=B^{-1}_{1/\Lambda}(p,\mathcal{M})
\int_{\partial\mathcal{B}_{1/\Lambda}(p,\mathcal{M})}\mathrm{d}^{n-1}p_1\,\phi(p_1),
\end{equation}
where the normalization factor is
\begin{equation}\label{sk-7-6}
B_{1/\Lambda}(p,\mathcal{M})=
\int_{\partial\mathcal{B}_{1/\Lambda}(p,\mathcal{M})}\mathrm{d}^{n-1}p_1.
\end{equation}
Considering the fact that the main term of the asymptotics at small distances for the function $G(\,\cdot\,,\,\cdot\,)$ has the same form as in the Euclidean case, see \eqref{sk-1-3}, then it can be argued that averaging with such a kernel is acceptable and the kernel belongs to the set $\Omega_0^\Lambda(\mathcal{M})$.
\begin{remark}\label{sk-r-12}
All definitions of this section can be extended similarly to the case of the submanifolds $\mathcal{M}_l$ and $\mathcal{M}_r$.
\end{remark}
\begin{lemma}\label{sk-l-9}
For each $\omega(\,\cdot\,,\,\cdot\,;\Lambda)\in\Omega_0^\Lambda(\mathcal{M})$ there are such $\omega_l(\,\cdot\,,\,\cdot\,;\Lambda)\in\Omega_0^\Lambda(\mathcal{M}_l)$ and $\omega_r(\,\cdot\,,\,\cdot\,;\Lambda)\in\Omega_0^\Lambda(\mathcal{M}_r)$ that for all $i\in\{l,r\}$, $p\in\mathcal{M}_{i,\Lambda}$, and $p_1\in\mathcal{M}_{i}$, the relation holds
\begin{equation}\label{sk-7-7}
\omega(p,p_1;\Lambda)-\omega_i(p,p_1;\Lambda)=0.
\end{equation}
\end{lemma}
\begin{proof}
Let $p,p_1\in\mathcal{M}_i$. Let us define functions on submanifolds as follows
\begin{equation}\label{sk-7-8}
\omega_i(p,p_1;\Lambda)=\omega(p,p_1;\Lambda)\bigg(
\int_{\mathcal{M}_i}\mathrm{d}^{n-1}p_2\,\omega(p,p_2;\Lambda)\bigg)^{-1},
\end{equation}
which, taking into account \eqref{sk-7-3}, leads only to a change in the normalization in the following domain $p\in\mathcal{M}_i\setminus\mathcal{M}_{i,\Lambda}$.
\end{proof}

\subsection{Regularization of functionals}
\label{sk:sec:reg:3}

In this section, a deformation of the classical action from Definition \ref{sk-d-6} is formulated. This approach, firstly, leads to the regularization of the perturbative decomposition for the quantum (effective) action from \eqref{sk-5-7}, and, secondly, is compatible with the process of gluing of manifolds and the corresponding partition functions, see Section \ref{sk:sec:app}.

First of all, we pay attention to some additional considerations regarding the introduction of regularization. It is clear that it should get rid of the divergences discussed at the end of Section \ref{sk:sec:def:3}, transforming them into functions depending on the regularizing parameter $\Lambda$ with singularities at $\Lambda\to+\infty$. However, a natural question arises: "At what stage should the deformation be introduced?". Our answer, based primarily on a large number of calculations in standard models from quantum field theory, is that the deformation should be introduced at the stage of determining the classical action. Various arguments can be given here, including the feasibility of the connection beetwen quantum equations of motion and quantum actions before and after the deformation, but we pay attention to only one specific reason. Although it is less formalized, at the same time it is strategic in nature.

As it is known, the functional integral does not currently have a strict mathematical description, except, of course, in some very simple cases. For this reason, when studying quantum field models, one has to work with perturbative decompositions, for which nothing else is required except the "normalization" and the ability to calculate standard Gaussian integrals. Given this, a general view of solving direct and inverse problems in mathematical physics can be extended to the functional integral method. We can undertand this transform as a kind of "black box" into which we sent signals and receive responses. Only, unlike classical physics problems, the role of the signal is played by a classical action (taking into account boundary conditions), and the role of the response is the correcponding quantum action. That is, we have the following diagram $S_{\mathrm{cl}}\to\blacksquare\to W_{\mathrm{eff}}$.

Using general logic to extend the properties of the classical integral to a wider class of objects, as well as using a comparison of the resulting quantum (effective) actions with experiments, in theoretical physics it was possible to construct a "copy"\footnote{The word "copy" means a mathematical version of the "black box", which produces the same results as "black box". In our case, from a combinatorial point of view.} for "black box", which, for well known classical models, produces "correct" effective actions. Based on such considerations, we can say that the resulting mathematical model was built empirically.

It is important to note here that for all known classical actions, the procedure for obtaining the corresponding effective actions is the same, up to a multiplier. This leads to the idea that in the process of introducing a regularization, while there is no strict mathematical description, the "black box" cannot be deformed. So, it is the classical action that needs to be changed. Moreover, as practice shows, the deformation of the functional integration process has to be selected separately for each model, while the process of introducing regularization through deforming classical actions applies to all models in the same way.

In addition, and this is very important, the introduction of regularization by deforming the process of the functional integration is tied to the use of perturbative decompositions, which thus become a significant limiting factor. This means that, after constructing the mathematical formalism for the functional integration, the issue related to the introduction of regularization will have to be solved anew, because models should be studied not only by decomposition in a small parameter.

\begin{remark}\label{sk-r-13}
A regularization of the effective action  $W_{\mathrm{eff}}$ should be introduced by deforming the classical action $S_{\mathrm{cl}}$.
\end{remark}

Next, the following question arises: "How exactly should the classical action be deformed?". This problem must be solved based on the need to preserve one or another property of the classical theory. At the same time, properties are understood not only as symmetries, which are important for describing physics, but also as a set of classical mathematical formalism that can be applied within the framework of theory. To clarify, let us look at an example of a quartic scalar theory in Euclidean space $\mathbb{R}^4$. In this case, an operator in the quadratic form and its fundamental solution in the momentum representation have the form
\begin{equation}\label{sk-8-1}
-\partial_{x_\mu}\partial_{x^\mu}+m^2,\,\,\,\frac{1}{|k|^2+m^2}.
\end{equation}
As is known, ultraviolet divergences appear due to the "slow" decrease of the fundamental solution at $|k|\to+\infty$, therefore one of the popular\footnote{The most popular regularization method is the dimensional one, which is achieved by transition to a space with non-integer dimension. Since it contains a number of significant open mathematical questions and is based, among other things, on a number of "agreements", it will not be used as an example in this paper.} deformations is the use of "higher" derivatives. It consists in the transition from the standard Laplace operator to a higher-order operator, for which the fundamental solution decreases "enough" quickly. For example,  it is possible to consider the following variant
\begin{equation}\label{sk-8-2}
\big(\partial_{x_\mu}\partial_{x^\mu}\big)/\Lambda^2-\partial_{x_\mu}\partial_{x^\mu}+m^2,\,\,\,\frac{1}{|k|^4/\Lambda^2+|k|^2+m^2}.
\end{equation}
In this case, the deformation is removed by the transition $\Lambda\to+\infty$. It is clear that from the point of view of regularization of Feynman diagrams, the desired result has been achieved. Nevertheless, in the process of such a transition, the opportunity to apply almost the entire formalism of classical mathematical physics concerning the standard Laplace operator was lost. It is important to note that the theory devoted to higher-order operators is much poorer. Moreover, in order to set the correct task for the received operator, additional conditions are needed, the use of which raises a question from a physical point of view. Additionally, we note that during the process of such deformation of a quadratic form, not only the operator itself is deformed, but also the "integration domain"\footnote{As a rule, this change is ignored, guided only by the desire to regularize the diagrams faster and to calculate the resulting integral.}$^{,}$\footnote{In the papers \cite{Iv-2024-1,Ivanov-Akac,Ivanov-Kharuk-2020,Ivanov-Kharuk-20222,Ivanov-Kharuk-2023,Iv-Kh-2024,Kh-2024}, such changes were taken into account, and the regularization used fully corresponded to all the considerations of this section.}.  This leads to an unnecessary\footnote{Unclear from a mathematical point of view.} changing the "black box", and that is bad. Therefore, one more remark can be formulated that narrows down the class of deformations.

\begin{remark}\label{sk-r-14}
A deformation of the classical action $S_{\mathrm{cl}}$ should not affect the operator in the quadratic form, that is, the functional $S_0$. Thus, only the term $S_{\mathrm{int}}$ responsible for interaction should be changed.
\end{remark}

In order to definitively formulate the rules, we conduct another discussion. Suppose that the operator in the quadratic form remains the same, and the fundamental solution has the form of \eqref{sk-8-2}. Let us rewrite it as follows
\begin{equation}\label{sk-8-3}
\frac{1}{|k|^4/\Lambda^2+|k|^2+m^2}=
\bigg(\frac{|k|^2+m^2}{|k|^4/\Lambda^2+|k|^2+m^2}\bigg)^{1/2}\frac{1}{|k|^2+m^2}\bigg(\frac{|k|^2+m^2}{|k|^4/\Lambda^2+|k|^2+m^2}\bigg)^{1/2}.
\end{equation}
The last partition means that the corresponding deformation can be achieved by changing the component $S_{\mathrm{int}}$ by averaging the fields according to the formula
\begin{equation}\label{sk-8-4}
\phi(x)\to\int_{\mathbb{R}^4}\mathrm{d}^4k\,e^{ik_\mu(x-y)^\mu}\bigg(\frac{|k|^2+m^2}{|k|^4/\Lambda^2+|k|^2+m^2}\bigg)^{1/2}\phi(y).
\end{equation}
In this case, the deformation is removed by the transition $\Lambda\to+\infty$, as it was before. This means that in the process of deformation, locality in the functional responsible for interaction is lost. Taking into account the previous remarks regarding the change of the "integration domain", we cannot provide examples of deformation without loss of locality. However, the type of non-locality can be controlled by making the theory almost local.
\begin{defa}\label{sk-d-21}
Let $\phi\in C^\infty(\mathcal{M},\mathbb{R})$, $p\in\mathcal{M}$, and $\Lambda>0$ is enough large. A deformation is called quasi-local if, during its introduction, the field $\phi(p)$ is transformed using the kernel of an integral operator $k(p,\,\cdot\,;\Lambda)$,
support of which is in a ball whose radius tends to zero as $\Lambda\to+\infty$ uniformly over the first argument.
\end{defa}
Now, after all the preliminary remarks, we proceed to the formulation of a regularization, which me and my co-author N.V.Kharuk have repeatedly and successfully applied in various theories in the Euclidean space $\mathbb{R}^n$, see \cite{34,Iv-2024-1,Ivanov-Akac,Ivanov-Kharuk-2020,Ivanov-Kharuk-20222,Ivanov-Kharuk-2023,Iv-Kh-2024,Kh-2024,Ivanov-2022,Iv-2024}. Since the main task in the process of previously performed calculations concerned the study of multi-loop integrals and comparison with known results, insufficient attention was paid to the discussion of the regularization process. It is for this reason that this section has been made so detailed.
\begin{defa}\label{sk-d-22}
Let $\Lambda>\Lambda_1$ taking into account Remarks \ref{sk-r-2} and \ref{sk-r-10}, $\omega\in\Omega_0^\Lambda(\mathcal{M})$, and $\mathrm{H}_\omega^\Lambda$ is the corresponding deforming operator from Definition \ref{sk-d-19}. Then, within the framework of this paper, the regularization of the effective action $W_{\mathrm{eff}}^{\phantom{1}}\to W_{\mathrm{eff}}^{\Lambda}$ consists in a deformation of the classical action $S_{\mathrm{cl}}^{\phantom{1}}\to S_{\mathrm{cl}}^{\Lambda}$ of the following type
\begin{equation}\label{sk-8-5}
S_{\mathrm{cl}}^{\phantom{1}}[\,\cdot\,,\mathcal{M}]=S_{\mathrm{0}}^{\phantom{1}}[\,\cdot\,,\mathcal{M}]+S_{\mathrm{int}}^{\phantom{1}}[\,\cdot\,,\mathcal{M}]\xrightarrow{\mbox{reg.}}
S_{\mathrm{cl}}^{\Lambda}[\,\cdot\,,\mathcal{M}]=
S_{\mathrm{0}}^{\phantom{1}}[\,\cdot\,,\mathcal{M}]+
S_{\mathrm{int}}^{\Lambda}[\,\cdot\,,\mathcal{M}],
\end{equation}
where
\begin{equation}\label{sk-8-6}
S_{\mathrm{int}}^{\Lambda}[\,\cdot\,,\mathcal{M}]=S_{\mathrm{int}}^{\phantom{1}}[\mathrm{H}_\omega^\Lambda(\,\cdot\,),\mathcal{M}_\Lambda].
\end{equation}
\end{defa}
\begin{col}\label{sk-c-7}
Taking into account Definition \ref{sk-d-21}, the deformation \eqref{sk-8-5} is quasi-local.
\end{col}
Returning to the discussions from Section \ref{sk:sec:def:3}, we can write out a perturbative decomposition for the regularized effective action in the form
\begin{equation}\label{sk-8-8}
	e^{-W_{\mathrm{eff}}^\Lambda[\sqrt{\hbar}\eta;\mathcal{M}]/\hbar}=e^{-S_{\mathrm{0}}[\phi^\eta;\mathcal{M}]}
	\bigg(
	S(\mathrm{H}_\omega^\Lambda(\delta_\psi+\phi^\eta);\mathcal{M}_\Lambda)e^{g(\psi;\mathcal{M})}\bigg)\bigg|_{\psi=0},
\end{equation}
where
\begin{multline}\label{sk-8-7}
	S(\mathrm{H}_\omega^\Lambda(\phi+\phi^\eta);\mathcal{M}_\Lambda)=
	1+\sum_{k\geqslant3}\hbar^{k/2-1}
	\Bigg(\prod_{i=1}^{k}\int_{\mathcal{M}_\Lambda}\mathrm{d}^np_{i}\Bigg)
	S_{k}(p_1,\ldots p_k)\times\\\times
	\mathrm{H}_\omega^\Lambda(\phi+\phi^\eta)(p_1)\ldots\mathrm{H}_\omega^\Lambda(\phi+\phi^\eta)(p_{k}),
\end{multline}
and the functional $g(\psi;\mathcal{M})$ is determined by relation \eqref{sk-5-5}. It can be seen from the explicit analytical representation that the regularization has been reduced to replacing fields with averaged analogues and deforming of the manifold from the integration domain $\mathcal{M}\to\mathcal{M}_\Lambda$.
\begin{lemma}\label{sk-l-10}
The decomposition from \eqref{sk-8-8} in each order by a small parameter $\sqrt{\hbar}$ does not contain divergences.
\end{lemma}
\begin{proof} The perturbative decomposition from \eqref{sk-8-8} for the effective action is obtained by applying the Wick's theorem, which actually cinsists in replacing all possible pairs $\phi(p_i)\phi(p_j)$ on the Green's function $G(p_i,p_j)$. Since in the regularized case, each field $\phi(p_i)$ is replaced by an averaged analog $\mathrm{H}_\omega^\Lambda(\phi)(p_i)$, then, after applying the Wick's theorem, we obtain the regularized functions of the form
\begin{equation}\label{sk-8-9}
G^\Lambda(p_i,p_j)=\int_{\mathcal{M}}\mathrm{d}^{n}\hat{p}_1\,\omega(p_i,\hat{p}_1;\Lambda)
\int_{\mathcal{M}}\mathrm{d}^{n}\hat{p}_2\,\omega(p_j,\hat{p}_2;\Lambda)G(\hat{p}_1,\hat{p}_2)=
\sum_\lambda\mathrm{H}_\omega^\Lambda(\psi_\lambda)(p_i)\frac{1}{\lambda}\mathrm{H}_\omega^\Lambda(\psi_\lambda)(p_j),
\end{equation}
where the decomposition from \eqref{sk-3-8} was used. As a result, the statement of the lemma follows from the facts that the kernel of the deforming operator  $\omega$ is from $\Omega_0^\Lambda(\mathcal{M})$, and the integration into \eqref{sk-8-7} occurs over the narrowed set (separated from the boundary).
\end{proof}
\begin{defa}\label{sk-d-23}
A regularized partition function, corresponding to the classical action $S_{\mathrm{cl}}[\,\cdot\,;\mathcal{M}]$, which is considered on a set of functions taking the value $\eta\in C^\infty(Y,\mathbb{R})$ on the boundary, is equal to
\begin{equation}\label{sk-5-10}
Z(\eta;\mathcal{M}_\Lambda)=e^{-W_{\mathrm{eff}}^\Lambda[\sqrt{\hbar}\eta;\mathcal{M}]/\hbar},
\end{equation}
where the boundary value $\eta$ and the deformed manifold $\mathcal{M}_\Lambda$ from the $S$-functional are highlighted as the main arguments. In addition, the partition function depends on the parameters of the classical action, the manifold $\mathcal{M}$, on which the boundary value problem is posed, and the deforming operator $\mathrm{H}_\omega^\Lambda$. However, the last arguments will usually be omitted.
\end{defa}
\begin{remark}\label{sk-r-15}
Formula \eqref{sk-8-8} can be similarly written for the submanifolds. From now on, we fix the choice of a kernel $\omega\in\Omega_0^\Lambda(\mathcal{M})$ for certainty. Also we select and fix kernels $\omega_l\in\Omega_0^\Lambda(\mathcal{M}_l)$ and $\omega_r\in\Omega_0^\Lambda(\mathcal{M}_r)$ on the submanifolds, taking into account Lemma \ref{sk-l-9}. The corresponding deformed Green's functions on $\mathcal{M}_l$ and $\mathcal{M}_r$ are notated as $G^\Lambda_l(\,\cdot\,,\,\cdot\,)$ and $G^\Lambda_r(\,\cdot\,,\,\cdot\,)$, respectively.
\end{remark}

\section{Application to the gluing process}
\label{sk:sec:app}

\subsection{Definitions and relations}
\label{sk:sec:app-1}
In this section, we assume that the requirements from the previous sections have been fulfilled. Consider a gluing of two manifolds $\mathcal{M}_l$ and $\mathcal{M}_r$ along the hypersurface $\Sigma$, as well as a transition from the regularized effective actions $W_{\mathrm{eff}}^\Lambda[\eta_l+\eta_\Sigma;\mathcal{M}_l]$ and $W_{\mathrm{eff}}^\Lambda[\eta_r+\eta_\Sigma;\mathcal{M}_r]$ and the corresponding partition functions, defined on the submanifolds $\mathcal{M}_l$ and $\mathcal{M}_r$, to the regularized effective action $W_{\mathrm{eff}}^\Lambda[\eta_l+\eta_r;\mathcal{M}]$ on the whole manifold $\mathcal{M}$ by an additional functional integration over all boundary functions $\eta_\Sigma$.

Considering the fact that the smooth submanifolds $\mathcal{M}_l$ and $\mathcal{M}_r$ were constructed by cutting the smooth manifold $\mathcal{M}=\mathcal{M}_l\cup_\Sigma\mathcal{M}_r$, then from a geometric point of view all the objects are well defined. Therefore, we can proceed immediately to functionals defined on these manifolds. There are two such functions: the classical action and the effective (quantum) one.

\begin{defa}\label{sk-d-17}
By gluing of two partition functions $Z(\eta_l+\eta_\Sigma;\mathcal{M}_{l,\Lambda})$ and $Z(\eta_r+\eta_\Sigma;\mathcal{M}_{r,\Lambda})$, corresponding to the effective actions $W_{\mathrm{eff}}^\Lambda[\sqrt{\hbar}(\eta_l+\eta_\Sigma);\mathcal{M}_l]$ and $W_{\mathrm{eff}}^\Lambda[\sqrt{\hbar}(\eta_r+\eta_\Sigma);\mathcal{M}_r]$, on the submanifolds $\mathcal{M}_l$ and $\mathcal{M}_r$ we call a functional integral of the form
\begin{align}\label{sk-6-1}
\tilde{Z}=&
\int_{\mathcal{H}_0(\Sigma)}\mathcal{D}\eta_\Sigma^{\phantom{1}}\,e^{-W_{\mathrm{eff}}^\Lambda[\sqrt{\hbar}(\eta_l+\eta_\Sigma);\mathcal{M}_l]/\hbar}
e^{-W_{\mathrm{eff}}^\Lambda[\sqrt{\hbar}(\eta_r+\eta_\Sigma);\mathcal{M}_r]/\hbar}\\\label{sk-6-21}=&
\int_{\mathcal{H}_0(\Sigma)}\mathcal{D}\eta_\Sigma^{\phantom{1}}\,
Z(\eta_l+\eta_\Sigma;\mathcal{M}_{l,\Lambda})
Z(\eta_r+\eta_\Sigma;\mathcal{M}_{r,\Lambda}),
\end{align}
which should be understood in the sense of a perturbative expansion in the small parameter $\sqrt{\hbar}$, defined by the equality
\begin{equation}\label{sk-6-2}
\bigg(e^{S_{\mathrm{0}}[\phi^{\delta_\eta}_l;\mathcal{M}_l]+S_{\mathrm{0}}[\phi^{\delta_\eta}_r;\mathcal{M}_r]}e^{-W_{\mathrm{eff}}^\Lambda[\sqrt{\hbar}(\eta_l+\delta_\eta);\mathcal{M}_l]/\hbar}
e^{-W_{\mathrm{eff}}^\Lambda[\sqrt{\hbar}(\eta_r+\delta_\eta);\mathcal{M}_r]/\hbar}e^{g_0(\eta;\Sigma)}\bigg)\bigg|_{\eta=0}.
\end{equation}
Here the quadratic form is defined by the formula
\begin{equation}\label{sk-6-6}
g_0(\eta;\Sigma)=\frac{1}{2}\int_{\Sigma}\mathrm{d}^{n-1}q_1\int_{\Sigma}\mathrm{d}^{n-1}q_2\,
\eta(q_1)G(q_1,q_2)\eta(q_2),
\end{equation}
and can be built using the kernel $G(q_1,q_2)$ of an integral operator, which is inverse operator to the one from the quadratic form $S_{\mathrm{0}}[\phi^{\eta}_l;\mathcal{M}_l]+S_{\mathrm{0}}[\phi^{\eta}_r;\mathcal{M}_r]$ according to Lemma \ref{sk-l-6}.
\end{defa}
\begin{remark}\label{sk-r-16}
The designations of the measure and the domain of integration in \eqref{sk-6-1} are symbolic, see Section \ref{sk:sec:def:3}.
\end{remark}
\begin{remark}\label{sk-r-17}
The first multiplier in formula \eqref{sk-6-2} subtracts the corresponding quadratic form from the effective actions, with the usage of which the averaging is carried out.
\end{remark}
\begin{col}\label{sk-c-8}
Taking into account Definitions \ref{sk-d-16} and \ref{sk-d-17}, the function from \eqref{sk-6-1} can be represented as
\begin{multline}\label{sk-6-8}
e^{-S_{\mathrm{0}}[\phi^{\eta_l}_l;\mathcal{M}_l]-S_{\mathrm{0}}[\phi^{\eta_r}_r;\mathcal{M}_r]}e^{S_{l,\Sigma}(\eta_l,\delta_\eta;\mathcal{M}_l)+S_{r,\Sigma}(\eta_r,\delta_\eta;\mathcal{M}_r)}\times\\\times
S(\mathrm{H}_\omega^\Lambda(\delta_{\psi_l}+\phi^{\delta_\eta}_l+\phi^{\eta_l}_l);\mathcal{M}_{l,\Lambda})
S(\mathrm{H}_\omega^\Lambda(\delta_{\psi_r}+\phi^{\delta_\eta}_r+\phi^{\eta_r}_r);\mathcal{M}_{r,\Lambda})\times\\\times
e^{g_l(\psi_l;\mathcal{M}_l)}
e^{g_r(\psi_r;\mathcal{M}_r)}
e^{g_0(\eta;\Sigma)}\Big|_{\psi_l=0,\,\psi_r=0,\,\eta=0},
\end{multline}
where the definitions for quadratic forms from the last line are written out according to Remark \ref{sk-r-8} and formula \eqref{sk-6-6}.
\end{col}

\begin{lemma}\label{sk-l-11}
Let $i,j\in\{l,r\}$, $p_1\in\mathcal{M}_{i,\Lambda}$, and $p_2\in\mathcal{M}_{j,\Lambda}$. Also let $\Lambda>\Lambda_1$ taking into account Remarks \ref{sk-r-2} and \ref{sk-r-10}.
Let us choose acceptable kernels $\{\omega,\omega_l,\omega_r\}$ for the deforming operator according to Remark \ref{sk-r-15}. Then the deformed Green's function from \eqref{sk-8-9} has the following representation 
\begin{equation}\label{sk-6-5}
G^\Lambda(p_1,p_2)=G_i^\Lambda(p_1,p_2)\delta_{ij}+
G^\Lambda_{ij}(p_1,p_2),
\end{equation}
where the auxiliary function is defined by the equality
\begin{multline}\label{sk-6-7}
G^\Lambda_{ij}(p_1,p_2)=
\int_{\Sigma}\mathrm{d}^{n-1}q_1\int_{\Sigma}\mathrm{d}^{n-1}q_2
\bigg(\int_{\mathcal{M}_i}\mathrm{d}^{n}p_3\,\omega_i(p_1,p_3;\Lambda)
N_{\Sigma,i}(q_1)G_i(p_3,q_1)\bigg)\times\\\times
G(q_1,q_2)
\bigg(\int_{\mathcal{M}_j}\mathrm{d}^{n}p_4\,\omega_j(p_2,p_4;\Lambda)
N_{\Sigma,j}(q_2)G_j(q_2,p_4)\bigg).
\end{multline}
\end{lemma}
\begin{proof} Let us use two well-known relations for Green's functions, see Theorem $2.1$ in \cite{sk-b-6} and Proposition $4.2$ in \cite{sk-16}. The first equality holds
\begin{equation}\label{sk-6-3}
G(p_1,p_2)=G_i(p_1,p_2)+
\int_{\Sigma}\mathrm{d}^{n-1}q_1\int_{\Sigma}\mathrm{d}^{n-1}q_2
\Big(N_{\Sigma,i}(q_1)G_i(p_1,q_1)\Big)
G(q_1,q_2)
\Big(N_{\Sigma,i}(q_2)G_i(q_2,p_2)\Big)
\end{equation}
for all $p_1,p_2\in\mathcal{M}_i$ such that $p_1\neq p_2$. It is performed outside the diagonal, since with matching arguments the right and left sides may have a singular behaviour. The second equality holds
\begin{equation}\label{sk-6-4}
G(p_1,p_2)=
\int_{\Sigma}\mathrm{d}^{n-1}q_1\int_{\Sigma}\mathrm{d}^{n-1}q_2
\Big(N_{\Sigma,l}(q_1)G_l(p_1,q_1)\Big)
G(q_1,q_2)
\Big(N_{\Sigma,r}(q_2)G_r(q_2,p_2)\Big)
\end{equation}
for all $p_1\in\mathcal{M}_l$ and $p_2\in\mathcal{M}_r$ such that $p_1\neq p_2$. Thus, the mentioned statement in \eqref{sk-6-5} follows from the last relations after applying the deforming operator $\mathrm{H}_\omega^\Lambda$ to them by both variables, taking into account equality \eqref{sk-7-7} from Lemma \ref{sk-l-9}.
\end{proof}

\subsection{The main statement}
\label{sk:sec:app-2}
\begin{theorem}\label{sk-t-1}
Taking into account all the above, the gluing of two regularized effective actions from \eqref{sk-6-1} is equal to 
\begin{equation}\label{sk-6-13}
\tilde{Z}=e^{-S_{\mathrm{0}}[\phi^{\eta_l+\eta_r};\mathcal{M}]}
S(\mathrm{H}_\omega^\Lambda(\delta_{\psi}+\phi^{\eta_l}+\phi^{\eta_r});\mathcal{M}_{l,\Lambda}\cup\mathcal{M}_{r,\Lambda})
e^{g(\psi;\mathcal{M})}\Big|_{\psi=0}=Z(\eta_l+\eta_r;\mathcal{M}_{l,\Lambda}\cup\mathcal{M}_{r,\Lambda}).
\end{equation}
At the same time, a gluing formula for the partition functions (the effective actions) takes the form
\begin{align}\label{sk-6-22}
Z(\eta_l+\eta_r;\mathcal{M}_\Lambda)&=
\lim_{\mathcal{M}_{l,\Lambda}\cup\mathcal{M}_{r,\Lambda}\to\mathcal{M}_\Lambda}
\int_{\mathcal{H}_0(\Sigma)}\mathcal{D}\eta_\Sigma^{\phantom{1}}\,
Z(\eta_l+\eta_\Sigma;\mathcal{M}_{l,\Lambda})
Z(\eta_r+\eta_\Sigma;\mathcal{M}_{r,\Lambda})\\\label{sk-6-27}
&\mathop{=}^{\mathrm{def}}\,
\Big\langle
Z(\eta_l+\,\cdot\,;\mathcal{M}_{l,\Lambda}),
Z(\eta_r+\,\cdot\,;\mathcal{M}_{r,\Lambda})\Big\rangle.
\end{align}
\end{theorem}
\begin{proof} It is necessary to show the equality of the formal series \eqref{sk-6-8} and \eqref{sk-6-13} with respect to the small parameter $\sqrt{\hbar}$. Starting with the first one, we convert it to the declared form. First of all, it is worth paying attention to the fact that the series from \eqref{sk-6-8} is obtained by applying the Wick's theorem over three fields and, therefore, is determined by field pairings, see \cite{sk-b-16,sk-29-3}. The main idea of this proof is to replace all possible combinations of pairings with equivalent ones, thus moving to a new simpler representation. 
	
\underline{Stage 1.} The first type of pairing corresponds to a situation where both fields belong to the functional $S(\,\cdot\,\,;\,\cdot\,)$. In this case, according to the Wick's theorem, there are two basic nonzero pairings
\begin{equation}\label{sk-6-9}
\delta_{\eta(q_1)}\delta_{\eta(q_2)}\to G(q_1,q_2),\,\,\,
\delta_{\psi_i(p_1)}\delta_{\psi_j(p_2)}\to G_{i}^{\phantom{'}}(p_1,p_2)\delta_{ij}^{\phantom{1}}.
\end{equation}
Given the fact that the boundary value of $\eta$ uniquely constructs a continuation inside the submanifold, see Lemma \ref{sk-l-4} and Definition \ref{sk-d-13}, the first transition can be rewritten in the following equivalent way
\begin{equation}\label{sk-6-10}
\phi^{\delta_{\eta}}_i(p_1)\phi^{\delta_{\eta}}_j(p_2)\to \int_{\Sigma}\mathrm{d}^{n-1}q_1\int_{\Sigma}\mathrm{d}^{n-1}q_2
\Big(N_{\Sigma,i}^{\phantom{1}}(q_1)G_i^{\phantom{'}}(p_1,q_1)\Big)
G(q_1,q_2)
\Big(N_{\Sigma,j}^{\phantom{1}}(q_2)G_j^{\phantom{'}}(q_2,p_2)\Big).
\end{equation}
Next, adding the deforming operators $\mathrm{H}_\omega^\Lambda$, we get
\begin{equation}\label{sk-6-11}
	\mathrm{H}_{\omega_i}^\Lambda(\phi^{\delta_{\eta}}_i)(p_1)\mathrm{H}_{\omega_j}^\Lambda(\phi^{\delta_{\eta}}_j)(p_2)\to
	G_{ij}^\Lambda(p_1,p_2),\,\,\,
\mathrm{H}_{\omega_i}^\Lambda(\delta_{\psi_i})(p_1)\mathrm{H}_{\omega_j}^\Lambda(\delta_{\psi_j})(p_2)\to
G_{i}^\Lambda(p_1,p_2)\delta_{ij}^{\phantom{1}}.
\end{equation}
The latter relation can be transformed by considering a linear combination of pairings. Indeed, taking into account the result of Lemma \ref{sk-l-11}, since the points $p_1$ and $p_2$ in the $S$-functionals belong to the deformed submanifolds, that is, $p_1\in\mathcal{M}_{i,\Lambda}$ and $p_2\in\mathcal{M}_{j,\Lambda}$, the following transition holds
\begin{equation}\label{sk-6-12}
\Big(\mathrm{H}_{\omega_i}^\Lambda(\delta_{\psi_i})(p_1)+\mathrm{H}_{\omega_i}^\Lambda(\phi^{\delta_{\eta}}_i)(p_1)\Big)\Big(\mathrm{H}_{\omega_j}^\Lambda(\delta_{\psi_j})(p_2)+\mathrm{H}_{\omega_j}^\Lambda(\phi^{\delta_{\eta}}_j)(p_2)\Big)\to G^\Lambda(p_1,p_2).
\end{equation}
This transition is important in our case, since the $S$-functionals depend precisely on the sum $\mathrm{H}_{\omega_i}^\Lambda(\delta_{\psi_i})+\mathrm{H}_{\omega_i}^\Lambda(\phi^{\delta_{\eta}}_i)$. Thus, taking into account Lemma \ref{sk-l-9}, it can be seen that the last transition can be represented in the equivalent way
\begin{equation}\label{sk-6-14}
\mathrm{H}_\omega^\Lambda(\delta_{\psi})(p_1)\mathrm{H}_\omega^\Lambda(\delta_{\psi})(p_2)\to G^\Lambda(p_1,p_2),
\end{equation}
meaning already the pairing over fields with the quadratic form $g(\psi;\mathcal{M})$, which corresponds to the problem on the manifold $\mathcal{M}$.

\underline{Stage 2.} The next type of pairing corresponds to the situation when the first field belongs to the functional $S(\,\cdot\,\,;\,\cdot\,)$, while the second one follows from the action on the boundary $S_{j,\Sigma}(\eta_j,\delta_\eta;\mathcal{M}_j)$, see \eqref{sk-3-13}. By conducting absolutely similar arguments, we make sure that the transition holds
\begin{multline*}\label{sk-6-15}
\mathrm{H}_{\omega_i}^\Lambda(\phi^{\delta_{\eta}}_i)(p_1)
S_{j,\Sigma}(\eta_j,\delta_\eta;\mathcal{M}_j)\to\\
\to
-\int_{\Sigma}\mathrm{d}^{n-1}q_1\int_{\Sigma}\mathrm{d}^{n-1}q_2
\bigg(\int_{\mathcal{M}_i}\mathrm{d}^{n}p_3\,\omega_i(p_1,p_3;\Lambda)N_{\Sigma,i}^{\phantom{1}}(q_1)G_i^{\phantom{'}}(p_3,q_1)\bigg)\times\\\times
G(q_1,q_2)\bigg(
\int_{Y_j}\mathrm{d}^{n-1}q_3
\Big(N_{\Sigma,j}^{\phantom{1}}(q_2)N_{j}^{\phantom{1}}(q_3)G_j^{\phantom{'}}(q_2,q_3)\Big)\eta_j(q_3)\bigg).
\end{multline*}
Note that the last integral can be rewritten using formulas \eqref{sk-6-3} and \eqref{sk-6-4} in the form
\begin{equation*}
-\int_{Y_j}\mathrm{d}^{n-1}q_3
\int_{\mathcal{M}_i}\mathrm{d}^{n}p_3\,\omega_i(p_1,p_3;\Lambda)
\Big(N_{j}^{\phantom{1}}(q_3)\big(G(p_3,q_3)-\delta_{ij}^{\phantom{'}}G_j^{\phantom{'}}(p_3,q_3)\big)\Big)\eta_j(q_3)=\mathrm{H}_{\omega_i}^\Lambda(\phi^{\eta_j}_{\phantom{1}}-\delta_{ij}\phi^{\eta_j}_{j})(p_1).
\end{equation*}
Therefore, given the inclusion $p_1\in\mathcal{M}_{i,\Lambda}$, the transition is valid for the sum of pairings
\begin{equation}\label{sk-6-17}
\mathrm{H}_{\omega_i}^\Lambda(\phi^{\delta_{\eta}}_i)(p_1)
\Big(S_{l,\Sigma}(\eta_l,\delta_\eta;\mathcal{M}_l)+S_{r,\Sigma}(\eta_r,\delta_\eta;\mathcal{M}_r)
\Big)
\to
\mathrm{H}_\omega^\Lambda(\phi^{\eta_l}_{\phantom{1}}+\phi^{\eta_r}_{\phantom{1}}-\phi^{\eta_i}_{i})(p_1).
\end{equation}
This means that the argument in $S$-functionals is shift to the last function, turning  $\phi^{\eta_i}_{i}$ into the sum $\phi^{\eta_l}_{\phantom{1}}+\phi^{\eta_r}_{\phantom{1}}$, constructed for the manifold $\mathcal{M}$. The last fact about the shift follows from a well-known auxiliary statement from quantum field theory, which, within the framework of the calculations under consideration, can be formulated as follows: let $F(\phi)$ be a functional decomposable in a series according to the argument $\phi$, then the shift of the argument can be rewritten as
\begin{equation}\label{sk-6-26}
F(\phi+\phi_0)=
\exp\bigg(\int_{\mathcal{M}}\mathrm{d}^np\,\phi_0(p)\delta_{\psi(p)}\bigg)
F(\phi+\psi)\bigg|_{\psi=0}.
\end{equation}
Such equality is proved by simple series expansion and calculation of derivatives. Finally, it remains only to pay attention that any pairing \eqref{sk-6-17} can be represented as an action of a functional derivative on an auxiliary field.

\underline{Stage 3.} The latter type of pairing includes only boundary actions, see \eqref{sk-3-13}. Following the general logic and using relations \eqref{sk-6-3} and \eqref{sk-6-4}, we write out the answer for the possible pair
\begin{multline*}\label{sk-6-18}
S_{i,\Sigma}(\eta_i,\delta_\eta;\mathcal{M}_i)
S_{j,\Sigma}(\eta_j,\delta_\eta;\mathcal{M}_j)
	\to
\int_{Y_i}\mathrm{d}^{n-1}q_1\int_{Y_j}\mathrm{d}^{n-1}q_2\,\eta_i(q_1)\times\\\times
\Big(N_{i}^{\phantom{1}}(q_1)N_{j}^{\phantom{1}}(q_2)
\big(G(q_1,q_2)-\delta_{ij}G_i(q_1,q_2)\big)\Big)\eta_j(q_2).
\end{multline*}
It is clear that the last sum, taking into account the results of Lemmas \ref{sk-l-5} and \ref{sk-l-7}, can be rewritten as $S_{l,r}(\eta_l,\eta_r;\mathcal{M})$, if $i\neq j$, and in the form
\begin{equation}\label{sk-6-19}
-2S_0[\phi^{\eta_j}_{\phantom{1}};\mathcal{M}]+2S_0[\phi^{\eta_j}_j;\mathcal{M}_j^{\phantom{'}}]
\end{equation}
for the case $i=j$. Finally, transferring the reasoning to the exponents and performing the partition into pairs for linear functions, we obtain the following relation
\begin{equation}\label{sk-6-20}
e^{S_{l,\Sigma}(\eta_l,\delta_\eta;\mathcal{M}_l)+S_{r,\Sigma}(\eta_r,\delta_\eta;\mathcal{M}_r)}
e^{g_0(\eta;\Sigma)}\Big|_{\eta=0}=e^{S_{l,r}(\eta_l,\eta_r;\mathcal{M})-S_0[\phi^{\eta_l};\mathcal{M}]-S_0[\phi^{\eta_r};\mathcal{M}]+S_0[\phi^{\eta_l}_l;\mathcal{M}_l]+S_0[\phi^{\eta_r}_r;\mathcal{M}_r]}.
\end{equation}

\underline{Stage 4.} After calculating all possible situations, it is necessary to rebuild the formula taking into account the proposed substitutions
\begin{multline}\label{sk-6-23}
	e^{-S_0[\phi^{\eta_l};\mathcal{M}]-S_0[\phi^{\eta_r};\mathcal{M}]}e^{S_{l,r}(\eta_l,\eta_r;\mathcal{M})}\times\\\times
	S(\mathrm{H}_\omega^\Lambda(\delta_{\psi}+\phi^{\eta_l}+\phi^{\eta_r});\mathcal{M}_{l,\Lambda})
	S(\mathrm{H}_\omega^\Lambda(\delta_{\psi}+\phi^{\eta_l}+\phi^{\eta_r});\mathcal{M}_{r,\Lambda})
	e^{g(\psi;\mathcal{M})}\Big|_{\psi=0}.
\end{multline}
Next, using Definition \eqref{sk-5-6} for the $S$-functional, we obtain the relation
\begin{equation}\label{sk-6-24}
S(\mathrm{H}_\omega^\Lambda(\delta_{\psi}+\phi^{\eta_l}+\phi^{\eta_r});\mathcal{M}_{l,\Lambda})
	S(\mathrm{H}_\omega^\Lambda(\delta_{\psi}+\phi^{\eta_l}+\phi^{\eta_r});\mathcal{M}_{r,\Lambda})=
S(\mathrm{H}_\omega^\Lambda(\delta_{\psi}+\phi^{\eta_l}+\phi^{\eta_r});\mathcal{M}_{l,\Lambda}\cup\mathcal{M}_{r,\Lambda}).
\end{equation}
Then, taking into account the decomposition from \eqref{sk-3-15}, equality \eqref{sk-8-8} and Definition \ref{sk-d-23}, formula \eqref{sk-6-23} turns into the declared equality from \eqref{sk-6-13}. 

\underline{Stage 5.} Relation \eqref{sk-6-22} is the result of a limit transition. Now we show that this transition is possible. Let $\theta(\cdot):[0,1]\to\mathcal{M}_\Lambda$ be a continuous (relative to the volume of the manifold) mapping of the line segment into a set of smooth manifolds such that
\begin{equation}\label{sk-6-25}
\theta(0)=\mathcal{M}_{l,\Lambda}\cup\mathcal{M}_{r,\Lambda},\,\,\,
\theta(1)=\mathcal{M}_{\Lambda}\,\,\,\mbox{and}\,\,\,
\theta(s_1)\subset\theta(s_2)
\end{equation}
for all $0\leqslant s_1<s_2\leqslant1$. It follows from the regularization construction that each coefficient in the expansion of the function $\Theta(s)=Z(\eta_l+\eta_r;\theta(s))$ by the parameter $\sqrt{\hbar}$ is regularized for all $s\in[0,1]$. Moreover, each coefficient of the expansion for $\Theta(\cdot)$ is a continuous function on $[0,1]$ due to the finite number of terms and the fact that $\theta(\cdot)$ plays the role of an integration domain in the integral of a continuous function. Thus, the transition is possible.
\end{proof}

\subsection{Corollaries}
\label{sk:sec:app-3}

\begin{col}\label{sk-c-9}
The statement of Theorem \ref{sk-t-1} is also true in the case when $\{t_k\}_{k=0}^2$ from \eqref{sk-2-5} are not equal to zero.
\end{col}
\begin{col}\label{sk-c-10}
The statement of Theorem \ref{sk-t-1}  remains true for the case when the coefficients $\{t_k\}$ depend on the point of the manifold (in a smooth way) and, moreover, have the form of a local differential operator. In the latter case, an averaging kernel of suitable smoothness should be used, see Definition \ref{sk-d-18}.
\end{col} 
\begin{proof}
In the process of proving the theorem, the internal structure for the coefficients of the functional $S_{\mathrm{int}}$ responsible for the interaction was not used.
\end{proof}
\begin{col}\label{sk-c-11}
If the model allows multiplicative renormalization, that is, the transition from the regularized action to a renormalized one can be carried out by redefining the constants of the classical action and adding local additional terms, then taking into account Corollary \ref{sk-c-10}, relation \eqref{sk-6-22} extends to the renormalized partition function (the quantum actions).
\end{col}
\begin{proof}
It is achieved by re-designation of the coupling constants and using the result of Corollary \ref{sk-c-10}.
\end{proof}

\section{Conclusion}
\label{sk:sec:dis}

In this paper, the regularization was constructed for the partition function corresponding to the scalar model on a smooth compact Riemannian manifold in general position and in arbitrary dimension. Such regularization is a generalization of the cutoff one in the coordinate representation and is introduced with the usage of a deforming operator in the classical action by averaging the fields in the part responsible for the interaction. The main Theorem \ref{sk-t-1} shows that this approach is consistent with the process of gluing of manifolds and the corresponding partition functions. Attention is also drawn to important corollaries, see Section \ref{sk:sec:app-3}, in particular, concerning renormalized partition functions.

\subsection{Remarks}

\hspace{\parindent}\textbf{On the dimension.}
Relation \eqref{sk-6-22} is true for any dimension. At the same time, \underline{it does not matter} whether the model is renormalizable or not. The equality \eqref{sk-6-13} is feasible regardless of how strong the singularities are contained in the partition function.

\vspace{2mm}
\textbf{On the mass parameter.} Taking into account Corollary \ref{sk-c-9}, the operator in \eqref{sk-3-1} can be defined for another value of the mass parameter $m_1^2$. With such a change, it is important to control only the condition that the first eigenvalue of the operator is positive. The difference of the mass parameters $m^2_{\phantom{1}}-m_1^2$ should be added to the term from \eqref{sk-2-5} responsible for the interaction.

\vspace{2mm}
\textbf{On the applicability of the regularization.} Separately, attention should be paid to the fact that the regularization formulated in Section \ref{sk:sec:reg:3} can be applied to a wider class of models, including vector, gauge and spinor fields. The recipe for introducing regularization is the same and simple: it is necessary to remove a neighborhood of the boundary and average the fields in the interaction using a $j$-acceptable kernel with a suitable value of $j$, see Definition \ref{sk-d-18}.

\vspace{2mm}
\textbf{On the generalization of Theorem to other models.} The gluing procedure for the regularized partition functions also allows a generalization. For example, in the case when the field has a set of indices, that is,  $\phi(p)\to\phi^a(p)$, and the operator is a matrix-valued one $A(p)\to A^{ab}(p)$. The main limiting factor in discussing this generalization is the possibility of correctly posed spectral problem on the manifold, similar to \eqref{sk-3-6}, on the cut surface, similar to Lemma \ref{sk-l-6}, as well as the solvability of the boundary value problem, as it was in \eqref{sk-3-9}. All these conditions affect only the free theory and are not related to the appearance of divergences.

\vspace{2mm}
\textbf{On the regularization of the determinant.} In addition, we comment on the topic related to the regularization of determinants. In Section \ref{sk:sec:def:3} it was indicated that the considered functional integral from \eqref{sk-5-1} contains such a normalization that in the absence of the interaction and under zero boundary conditions, the integral is equal to one. Such a requirement is natural, including from a physical point of view. Indeed, in the absence of external signals, a free system is of no interest.

At the same time, in formula \eqref{sk-5-1}, the determinant, or rather a division of two determinants, is well defined. Let $\mathcal{A}(\eta)$ denote the operator \eqref{sk-3-1} on the manifold $\mathcal{M}$, given on the set of functions taking on the boundary a value  $\eta\in C^\infty(\mathcal{M},\mathbb{R})$. Then, in the absence of interaction, the integral from \eqref{sk-5-1} is rewritten as
\begin{align*}
	e^{-W_{\mathrm{eff}}[\sqrt{\hbar}\eta;\mathcal{M}]/\hbar}\Big|_{\{t_k=0\}_{k\geqslant0}}=&\,\,
	\mathcal{N}^{-1}(\mathcal{M})
	\int\limits_{\mathcal{H}(\sqrt{\hbar}\eta;\mathcal{M})}
	\mathcal{D}\phi\,e^{-S_{\mathrm{0}}[\phi;\mathcal{M}]/\hbar}\\
	=&	\int\limits_{\mathcal{H}(\eta;\mathcal{M})}
	\mathcal{D}\phi\,e^{-S_{\mathrm{0}}[\phi;\mathcal{M}]}\bigg/
		\int\limits_{\mathcal{H}(0;\mathcal{M})}
	\mathcal{D}\phi\,e^{-S_{\mathrm{0}}[\phi;\mathcal{M}]}\\
	=&\,\,\sqrt{\frac{\det(\mathcal{A}(0))}{\det(\mathcal{A}(\eta))}}
	=e^{-S_0[\phi^{\eta};\mathcal{M}]}.
\end{align*}
Thus, the relation is defined correctly. Nevertheless, we note that the determinants for operators with zero boundary conditions, depending only on the geometry of the manifold, taking into account a suitable regularization, also satisfy the gluing relation, see \cite{sk-b-17,sk-b-18} and Section $4.1$ in \cite{sk-16}. In the framework of this work devoted to quantum field issues, this topic was not of interest and was not considered, since the divergences for the "empty" problem are not related to "standard" ultraviolet divergences and are the result of an incorrect choice of normalization.

\vspace{2mm}
\textbf{On the issue of renormalization.} This item is devoted to discussing the type of divergences that arise in the partition function. The fact is that the regularized partition function can be understood in two different ways. Either as a functional defined on the fields $\eta$ on the boundary, that is, in the form of a formal series
\begin{equation}\label{sk-9-1}
Z(\eta;\mathcal{M}_\Lambda)=e^{-S_0[\phi^{\eta};\mathcal{M}]}\Bigg(1+\sum_{k=1}^{+\infty}
\bigg(\prod_{i=1}^{k}\int_{Y}\mathrm{d}^{n-1}q_i\bigg)Z_{1,k}(q_1,\ldots,q_2)\eta(q_1)\ldots\eta(q_k)\Bigg),
\end{equation}
or as a functional defined on the fields $\phi^{\eta}$, having a given behavior on the boundary and extended inside the manifold $\mathcal{M}$, that is, in the form
\begin{equation}\label{sk-9-2}
	Z(\eta;\mathcal{M}_\Lambda)=e^{-S_0[\phi^{\eta};\mathcal{M}]}\Bigg(1+\sum_{k=1}^{+\infty}
	\bigg(\prod_{i=1}^{k}\int_{\mathcal{M}_\Lambda}\mathrm{d}^np_i\bigg)Z_{2,k}(p_1,\ldots,p_2)\phi^{\eta}(p_1)\ldots\phi^{\eta}(p_k)\Bigg).
\end{equation}
In this case, representation \eqref{sk-9-1} is obtained from \eqref{sk-9-2} by substituting \eqref{sk-3-10}. Thus, the densities $Z_{1,k}$ and $Z_{2,k}$ are related to each other. However, the transition for a multiplier with a free action is not equivalent. Indeed, the partitions function according to formula \eqref{sk-8-8} contains an action of the form  $S_0[\phi^\eta;\mathcal{M}]$, which according to the decomposition from \eqref{sk-3-5} contains the functionals \eqref{sk-3-20}. The latter objects are well defined, however, if we substitute an explicit form of the integral kernel for the operators from \eqref{sk-3-19}, then non-integrable densities will arise. In this case, the simplest well-defined functional will lose convergence, and additional divergences will appear in theory.

Given the mutual ambiguity of  $\eta\leftrightarrow\phi^\eta$, it is proposed to consider the decomposition from \eqref{sk-9-2} as the correct representation. In this case, the renormalization process takes on a more elegant form. Indeed, with the use of \eqref{sk-9-1}, additional divergences appeared during gluing and the renormalization process was not clear, because some additional renormalization procedure was needed that corresponds to the gluing mapping. 

If we use \eqref{sk-9-2}, everything turns out neatly. It is sufficient to have a multiplicative renormalization of the model on a compact smooth manifold (without discussing gluing). Then, from the consistency of regularization with the gluing process, it follows that the renormalized objects are also consistent with it, and the gluing process itself does not introduce new divergences.

\subsection{Acknowledgements}

\hspace{\parindent}The work is supported by the Ministry of Science and Higher Education of the Russian
Federation in the framework of a scientific project under agreement from 14.06.2024, grant № 075-15-2024-631. Also, the author is the winner of the "Young Mathematics of Russia" contest conducted by the Theoretical Physics and Mathematics Advancement Foundation "BASIS" and would like to thank its sponsors and jury.

\vspace{2mm}
The author thanks N.Yu.Reshetikhin for recommending the link \cite{sk-16}. The author is grateful to N.V.Kharuk for numerous useful comments and for working together on related topics. The author is also grateful to K.A.Ivanov for creating stimulating conditions and a delightful smile.

\vspace{2mm}
\textbf{Historical background to the title page.} Nikolai Maksimovich Kalita was born on December 19, 1925 in the village of Karkhovka, Chernihiv Oblast of the Ukrainian SSR in a peasant family. His parents were natives of the same village Maxim Titovich Kalita, who was born on January 21, 1899 and died during the Great Patriotic War at the front, and Domna Vasilyevna Kalita (Kostenko), who was born on October 10, 1900 and met old age in Girmanka. Their lineage can be traced back to the metrical books relating to the Ioanno-Bogoslovskaya Church. According to the entries in N.M.'s work book, he served in the Soviet Army from December 11, 1944 to October 13, 1978 and retired with the rank of major. He married Lydia Matveevna Mukhopad on September 2, 1951 in Lviv, who was born on March 18 (April 17, according to documents), 1929 in the khutor of Babakovka in the Bilopillia Raion of the Sumy Oblast. Later, they had two children: Svetlana (28.01.1953, Horodok, Lviv Oblast) and Natalia (06.05.1956, Lutsk, Volyn Oblast). After retiring in the 80s, he and his family moved to Donetsk, where he lived until his death on December 8, 2000. His wife L.M. also died in Donetsk on February 9, 2015.

Nikolai M. Kalita had a significant influence on the author's worldview and is an example of an intelligent, decent and kind person for him.

\vspace{2mm}
\textbf{Data availability statement.} Data sharing not applicable to this article as no datasets were generated or analysed during the current study.

\vspace{2mm}
\textbf{Conflict of interest statement.} The author states that there is no conflict of interest.

\end{document}